\documentclass[journal,onecolumn]{IEEEtran}
\usepackage{graphicx}
\usepackage{amsmath}
\usepackage{booktabs}
\usepackage{latexsym}
\usepackage{mathrsfs}
\usepackage{amsthm}
\usepackage{amssymb}
\usepackage{amsfonts}
\usepackage{amsbsy}

\usepackage[shortlabels]{enumitem}
\usepackage{url}
\usepackage{array}
\usepackage{pdflscape}
\usepackage{xcolor}
\usepackage{stmaryrd}
\usepackage{verbatim}

\newcommand{\Ff}{{\mathbb F}}

\newcommand{\Cc}{{\mathbb C}}

\newcommand\cc{{\mathcal C}}        %

\newcommand{\F}{ {{\mathbb F}} }

\newcommand\cF{{\mathbf c}}
\def\Tr{\operatorname{Tr}}

\theoremstyle{plain}% default
\newtheorem{thm}{Theorem}

\newtheorem{prop}[thm]{Proposition}

\newtheorem{cor}{Corollary}
\theoremstyle{definition}
\newtheorem{defn}{Definition}
\newtheorem{example}{Example}
\newtheorem{remark}{Remark}
\usepackage[colorlinks=true]{hyperref}
\begin{document}
\title{Improved Spectral Bound for Quasi-Cyclic Codes}

%%%
\author{Gaojun Luo, Martianus Frederic Ezerman, San Ling, and Buket \"{O}zkaya
\thanks{The work of Gaojun Luo, Martianus Frederic Ezerman, and San Ling is supported by Nanyang Technological University Research Grant 04INS000047C230GRT01. The work of San Ling is supported in part by the National Natural Science Foundation of China under Grant 11971175.}
\thanks{Gaojun Luo, Martianus Frederic Ezerman, and San Ling are with the School of Physical and Mathematical Sciences, Nanyang Technological University, Singapore 637371, (e-mails: $\{\rm gaojun.luo, fredezerman, lingsan\}$@ntu.edu.sg).}
%\thanks{Martianus Frederic Ezerman is with the School of Physical and Mathematical Sciences, Nanyang Technological University, Singapore 637371, and also with Sandhiguna, Graha Pena, Kota Batam, Kepulauan Riau 29461, Indonesia (e-mail: fredezerman@ntu.edu.sg).}
\thanks{Buket \"{O}zkaya is with the Institute of Applied Mathematics, Middle East Technical University, 06800 Ankara, Turkey (e-mail: ${\rm ozkayab}$@metu.edu.tr).}
%\thanks{Copyright (c) 2023 IEEE. Personal use of this material is permitted. However, permission to use this material for any other purposes must be obtained from the IEEE by sending a request to pubs-permissions@ieee.org.}
}

%\markboth{}%
%{Shell \MakeLowercase{\textit{et al.}}: Bare Demo of IEEEtran.cls for IEEE Transactions on Magnetics Journals}

\maketitle

\begin{abstract}
Spectral bounds form a powerful tool to estimate the minimum distances of quasi-cyclic codes. They generalize the defining set bounds of cyclic codes to those of quasi-cyclic codes. Based on the eigenvalues of quasi-cyclic codes and the corresponding eigenspaces, we provide an improved spectral bound for quasi-cyclic codes. Numerical results verify that the improved bound outperforms the Jensen bound in almost all cases. Based on the improved bound, we propose a general construction of quasi-cyclic codes with excellent designed minimum distances. For the quasi-cyclic codes produced by this general construction, the improved spectral bound is always sharper than the Jensen bound.
\end{abstract}


\maketitle
%\tableofcontents

\section{Introduction}\label{sec:intro}

Quasi-cyclic codes constitute an important class in coding theory as they possess nice algebraic structures and have been widely applied. Chen, Peterson, and Weldon in \cite{Chen1969} asserted the existence of asymptotically good quasi-cyclic codes. Numerous quasi-cyclic codes, {\it e.g.}, those found and recorded in \cite{Gulliver1991,Grassl}, have the respective best-known minimum Hamming distances among comparative linear codes of fixed length and dimension. In convolutional coding schemes, quasi-cyclic codes were utilized in \cite{Solomon1979,Esmaeili1998,Lally2006} to characterize convolutional codes. Low Density Parity-Check (LDPC) codes with excellent performance parameters, which serve as a crucial instrument in modern communication systems, were constructed from quasi-cyclic codes in \cite{Lan2007,Kang2010}.

In contrast to the case of cyclic codes, estimates on the minimum distances of quasi-cyclic codes have not been explored as extensively. Finding good lower bounds on the minimum distances of quasi-cyclic codes can be traced back to the work of Jensen in \cite{Jensen1985}, where a bound was derived based on the concatenated structure of the codes. Lally in \cite{Lally2003} established another lower bound by using the index and co-index of the codes.

For cyclic codes, the \emph{zero sets} of the codes provide valuable information on their minimum distances. Well-known bounds, {\it e.g.}, the Bose-Chauduri-Hocquenghem (BCH) bound and the Hartmann-Tzeng (HT) bound, follow from the zero sets, leading to such bounds being classified into \emph{defining set bounds} of cyclic codes. Turning to quasi-cyclic codes, Semenov and Trifonov in \cite{Semenov2012} established the spectral theory of quasi-cyclic codes based on the eigenvalues of polynomial matrices. The same work derived the BCH-like bound. In \cite{Roth2019}, Roth and Zeh followed up by constructing quasi-cyclic codes with designed minimum distances based on the BCH-like bound. Zeh and Ling in \cite{Zeh2016} extended the results of \cite{Semenov2012} by coming up with the HT-like bound and putting forward the spectral theory of quasi-cyclic product codes. In a more recent work, Ezerman {\it et al.} in \cite{Ezerman2021} showed that \emph{any defining set bound of cyclic codes is applicable to quasi-cyclic codes}. This general bound, which is also called the \emph{spectral bound}, includes the BCH-like bound in \cite{Semenov2012} and the HT-like bound in \cite{Zeh2016} as special cases. The numerical comparisons supplied in \cite{Ezerman2021} highlighted that the Jensen bound outperforms the Lally bound and the spectral bound in most cases.

\smallskip
\noindent
\textbf{Our Contributions}

In this paper, we propose a new lower bound on the minimum distances of quasi-cyclic codes based on their spectral theory. We summarize the contributions as follow.
\begin{enumerate}[leftmargin=*]
\item Based on the eigenvalues and the corresponding eigenspaces of any quasi-cyclic code $\cc$, we build an improved spectral bound for $\cc$ in Theorem \ref{ISB}. This bound refines the spectral bound of \cite{Ezerman2021}.  We prove that our improved spectral bound is tighter than that of \cite{Ezerman2021}. The improvement is further confirmed by numerical results in Table \ref{table:simulation}. Let the length and index of $\cc$ be, respectively, $\ell n$ and $\ell$. The spectral bounds in \cite{Ezerman2021} cannot exceed $\max\{n,\ell\}$, whereas our improved spectral bound is greater than $\max\{n,\ell\}$ under certain conditions. Example \ref{example1} exhibits this breakthrough. By constructing random binary and ternary quasi-cyclic codes and comparing the relevant bounds, we obtain numerical validation that our improved spectral bound beats the Jensen bound in almost all cases. Combined with the results in \cite{Ezerman2021}, our improved spectral bound outperforms the Jensen bound of \cite{Jensen1985}, the Lally bound from  \cite{Lally2003}, and the spectral bounds proposed in \cite{Ezerman2021}.

\item Based on our improved spectral bound, we provide a general construction of quasi-cyclic codes with designed minimum distances. Quasi-cyclic codes with small Singleton defects can then be constructed. Example \ref{Example1:QC-BCH} showcases that the dimensions of quasi-cyclic codes from this general construction can be greater than those of the BCH codes for a given length and designed minimum distances. We subsequently prove that the improved spectral bound for quasi-cyclic codes produced by this general construction is always tighter than the corresponding Jensen bound.

\item Locally repairable codes have become a hot topic due to data explosion in the past decades. Sufficient conditions for quasi-cyclic codes to be locally repairable can be formulated from their concatenated structure or from suitable generator polynomial matrices. On the other hand, upper bounds on the minimum distances of quasi-cyclic codes can be obtained from the bounds for locally repairable codes. We use the Jensen bound or the improved spectral bound to propose three constructions of quasi-cyclic codes whose distances reach the upper bounds.
\end{enumerate}

\smallskip
\noindent
\textbf{Organization}

After this introduction, Section \ref{S2} reviews basic definitions and results about cyclic codes and quasi-cyclic codes. We propose our improved spectral bound in Section \ref{S3}. Section \ref{S4} compares the improved spectral bound against the Jensen bound and the previous spectral bound in terms of performance. We provide a general construction of quasi-cyclic codes with designed minimum distances in Section \ref{S5}. Sections \ref{S6} presents the connection between quasi-cyclic codes and locally repairable codes. The concluding remarks in Section \ref{S7} wrap the paper up. All computations, including the examples and the simulation results in Section \ref{S4}, were done in {\tt MAGMA} V2.27-8 \cite{Bosma1997}.

\section{Preliminaries}\label{S2}

This section collects some basic results about cyclic codes, quasi-cyclic codes, and minimum distance bounds from the roots of their generator polynomials or the eigenvalues of generator polynomial matrices. We write vectors and polynomial matrices in bold font. Matrices with scalar entries are written in normal font.

Let $q$ be a prime power and let $\Ff_q$ stand for the finite field with $q$ elements. Given a positive integer $m$, we use $[m]$ to denote the set $\{1,2,\cdots,m\}$. An $[n,k,d(\cc)]_q$ linear code $\cc$ is a $k$-dimensional linear subspace of $\Ff_q^n$ with minimum Hamming distance $d(\cc)$. The code $\cc$ is \textit{maximum distance separable} (MDS) if $d(\cc)=n-k+1$. The \textit{Singleton defect} of a code is the difference $n-k+1-d(\cc)$. Taking the alphabet size of $\cc$ into consideration, the Griesmer bound in \cite{Griesmer1960} tells us that
\[
n\geq \mathcal{G}(k,d(\cc))=\sum_{i=0}^{k-1}\left\lceil\frac{d(\cc)}{q^i}\right\rceil.
\]
If $k=0$, then $\cc=\{\mathbf{0}\}$ is \emph{trivial} and, throughout this paper, its minimum distance is $d=\infty$.

\subsection{Cyclic Codes}

Let $n>1$ be, throughout, an integer with $\gcd(n,q)=1$. A linear code $\cc$ is a \emph{cyclic} code if $(c_{n-1},c_0,\cdots,c_{n-2})\in\cc$ for each codeword $(c_0,c_1,\cdots,c_{n-1})\in\cc$. In other words, the code $\cc$ is invariant under a cyclic shift of codewords. A codeword $\mathbf{c}=(c_0,c_1,\cdots,c_{n-1}) \in \cc$ can be expressed as a polynomial $c(x)=c_0+c_1x+\cdots+c_{n-1}x^{n-1}\in\Ff_q[x]$. This expression allows us to identify the cyclic code $\cc$ as an ideal in the residue class ring $R=\Ff_q[x]/\langle x^n-1\rangle$. Each ideal of $R$ is principal. There exists a unique monic polynomial $g(x)\in R$ of degree $n-k$ such that $\cc=\langle g(x)\rangle$. The polynomial $g(x)$ divides $x^n-1$ and is called the \emph{generator polynomial} of $\cc$, whereas the polynomial $h(x)\in R$ with $h(x) \, g(x)=x^n-1$ is called the \emph{check polynomial} of $\cc$.

Let $r$ be the smallest positive integer such that $q^r\equiv 1 \pmod n$ and let $\alpha$ be a primitive $n^{\rm th}$ root of unity in $\Ff_{q^r}$. The set of all roots of $x^n-1$ is $\Delta=\{\alpha^i:0\leq i\leq n-1\}$. Given a cyclic code $\cc=\langle g(x)\rangle$ of length $n$, the \emph{zero set} of $\cc$ is $Z=\{\alpha^i: 0\leq i\leq n-1, g(\alpha^i)=0\}$. The power set $\mathcal{P}(Z)$ of $Z$ is the \emph{defining set} of $\cc$. Several lower bounds on $d(\cc)$ have been established in \cite{Bose1960,hocquenghem1959codes,Hartmann1972,Roos1983,Lint1986,Betti2006} by selecting a specific defining set of $\cc$. By closely studying these defining set bounds, we formulate the following definition.
\begin{defn}\label{def1}
Let $\cc$ be a cyclic code with zero set $Z$ and minimum distance $d(\cc)$. A \emph{defining set bound} of $\cc$ is an element of some set $D(\cc)=\{(L,d_L):L\subseteq Z,d_L\in(\mathbb{N}\cup\{\infty\}) \mbox{, where } d(\cc)\geq d_L\}$.
\end{defn}

A \emph{consecutive set} is defined to be $\{\alpha^{j+mi}: 0\leq i\leq \delta-2 \}$ with $j\geq 0$, $2\leq\delta\leq n+1$, $m>0$, and $\gcd(m,n)=1$. The BCH bound in \cite{Bose1960,hocquenghem1959codes} corresponds to $\{(L,|L|+1):L\subseteq Z \mbox{ is consecutive}\}$. The respective generalizations of the BCH bound in \cite{Hartmann1972,Roos1983,Lint1986,Betti2006} can be formulated according to Definition \ref{def1}.

\subsection{Quasi-Cyclic Codes and Their Spectral Bound}

Let $\ell$ be a positive integer. An $[\ell n,k,d]_q$ linear code $\cc$ is a \emph{quasi-cyclic} code of \textit{index} $\ell$ and \textit{co-index} $n$ if the code is invariant under a cyclic shift of codewords by $\ell$ positions with $\ell$ being the smallest number satisfying this property. If we arrange the entries of a codeword $c\in\cc$ to form an $n\times \ell$ array
\begin{equation}\label{s2eq1}
c=\begin{pmatrix}
c_{0,0}&c_{0,1}&\cdots&c_{0,\ell-1}\\
c_{1,0}&c_{1,1}&\cdots&c_{1,\ell-1}\\
\vdots& \vdots &  \ddots & \vdots\\
c_{n-1,0}&c_{n-1,1}&\cdots&c_{n-1,\ell-1}
\end{pmatrix},
\end{equation}
then the code $\cc$ is invariant under the reordering of the rows through a cyclic shift. Since $R=\Ff_q[x]/\langle x^n-1\rangle$, the codeword $c$ is associated with the polynomial vector
\[
\mathbf{c}(x)=(c_0(x),c_1(x),\cdots,c_{\ell-1}(x))\in R^\ell,
\]
where $c_i(x)=c_{0,i} + c_{1,i} \,x + \ldots + c_{n-1,i} \, x^{n-1} \in R$ for each $0\leq i\leq \ell-1$. This representation allows for the identification of $\cc$ as an $R$-submodule of $R^\ell$ and we can write
\[
\cc=\left\{\mathbf{a}(x) \, \mathbf{G}(x)\in R^\ell: \mathbf{a}(x)\in(\Ff_q[x])^\ell \right\},
\]
where $\mathbf{G}(x)$ is an $\ell\times \ell$ polynomial matrix over $\Ff_q[x]$. The matrix $\mathbf{G}(x)$ is a \emph{generator polynomial matrix} of $\cc$. It was proved in \cite{Lally2001} that $\mathbf{G}(x)$ can be reduced to the simplified form
\begin{equation}\label{PGM1}
\widetilde{\mathbf{G}}(x)=\begin{pmatrix}
g_{0,0}(x)&g_{0,1}(x)&\cdots&g_{0,\ell-1}(x)\\
0&g_{1,1}(x)&\cdots&g_{1,\ell-1}(x)\\
\vdots& \vdots &  \ddots & \vdots\\
0&0&\cdots&g_{\ell-1,\ell-1}(x)
\end{pmatrix}
\end{equation}
that satisfies the following four conditions:
\begin{enumerate}
\item $\widetilde{\mathbf{G}}(x)$ is upper-triangular, which means that $g_{i,j}(x)=0$ for $i>j$.
\item $\deg(g_{i,j}(x))<\deg(g_{j,j}(x))$ for each $i<j$.
\item $g_{j,j}(x)\mid (x^n-1)$ with $g_{j,j}(x)\neq 0$ for each $0\leq j\leq \ell-1$.
\item $g_{i,j}(x)=0$ for each $i\neq j$ if $g_{j,j}(x)=x^n-1$.
\end{enumerate}
The dimension $k$ of $\cc$ is $\ell n-\sum_{j=0}^{\ell-1}\deg(g_{j,j}(x))$. Since the matrix $\widetilde{\mathbf{G}}(x)$ is obtained by performing elementary operations on rows of $\mathbf{G}(x)$, there exists a generator polynomial matrix $\mathbf{G}^{\prime}(x)$ of the form \eqref{PGM1} when the second and fourth conditions are removed. In fact, the matrix $\mathbf{G}^{\prime}(x)$ admits the following stronger condition, which implies $(1)$ and $(3)$.
\begin{enumerate}
\item[(5)] There exists an upper-triangular $\ell\times\ell$ polynomial matrix $\mathbf{H}(x)$ over $\Ff_q[x]$ such that $\mathbf{H}(x)\mathbf{G}^{\prime}(x)=(x^n-1)\mathbf{I}_\ell$ with $\mathbf{I}_\ell$ being the identity matrix of order $\ell$.
\end{enumerate}
As shown in \cite{Lally2001}, the matrix $\mathbf{G}^{\prime}(x)$ can be reduced to the matrix $\widetilde{\mathbf{G}}(x)$ that satisfies Conditions $(1)-(4)$ by elementary operations on the rows. In particular, $\mathbf{G}^{\prime}(x)$ and $\widetilde{\mathbf{G}}(x)$ have the same diagonal entries. When a quasi-cyclic code has such a generator polynomial matrix $\mathbf{G}^{\prime}(x)$, its dimension is still $\ell n-\sum_{j=0}^{\ell-1}\deg(g_{j,j}(x))$. Conversely, if an $\ell\times\ell$ polynomial matrix over $\Ff_q[x]$ meets Condition $(5)$, then it generates a quasi-cyclic code of index $\ell$ and length $\ell n$.

Based on the polynomial matrix $\widetilde{\mathbf{G}}(x)$ in \eqref{PGM1}, Semenov and Trifonov in \cite{Semenov2012} introduced a spectral bound on the minimum distance of quasi-cyclic codes. This result has been further developed in \cite{Zeh2016,Ezerman2021}. The determinant of $\widetilde{\mathbf{G}}(x)$ is $\det(\widetilde{\mathbf{G}}(x))=\prod_{j=0}^{\ell-1} g_{j,j}(x)$. The roots of $\det(\widetilde{\mathbf{G}}(x))$ are called \emph{eigenvalues} of the quasi-cyclic code $\cc$. The \emph{algebraic multiplicity} of an eigenvalue $\alpha$ is the largest integer $s$ such that $(x-\alpha)^s$ divides $\det(\widetilde{\mathbf{G}}(x))$. Let $r$ be the smallest positive integer such that $q^r\equiv 1 \pmod n$. The \emph{eigenspace} $V_\alpha$ of $\alpha$ is defined as
\begin{equation}\label{eq:eigenspace}
V_\alpha=\{\mathbf{u}\in\Ff_{q^r}^\ell:\widetilde{\mathbf{G}}(\alpha)\mathbf{u}^{\top}=\mathbf{0}^{\top}\}
\end{equation}
and its $\Ff_{q^r}$-dimension is the \emph{geometric multiplicity} of $\alpha$. Semenov and Trifonov asserted in \cite{Semenov2012} that the algebraic multiplicity of any eigenvalue of $\cc$ is equal to its geometric multiplicity. Note that $\Delta=\{\alpha^i:0\leq i\leq n-1\}$ is the set of roots of $x^n-1$. Let $E\subseteq \Delta$ be the set of all eigenvalues of $\cc$. If the diagonal elements $g_{j,j}(x)$ of $\widetilde{\mathbf{G}}(x)$ are constant polynomials, then $E$ is the empty set and the code $\cc$ has parameters $[\ell n,\ell n,1]_q$. If $E=\Delta$ and each eigenvalue of $E$ has geometric multiplicity $\ell$, then $\cc=\{\mathbf{0}\}$. We label the elements of $E$ as $\{\alpha^{i_1},\cdots,\alpha^{i_t}\}$. Let $m_j$ be the geometric multiplicity of $\alpha^{i_j}$ and let $\{\mathbf{v}_{j,0},\cdots,\mathbf{v}_{j,m_j-1}\}$ be a basis of its eigenspace $V_{\alpha^{i_j}}$ for each $1\leq j\leq t$. We define an $m_j\times \ell$ matrix $P_{i_j}$ whose rows are formed by $\mathbf{v}_{j,0},\cdots,\mathbf{v}_{j,m_j-1}$ and an $m_j \times \ell n$ matrix
\[
H_j=(\alpha^{i_j\cdot0},\cdots,\alpha^{i_j(n-1)})\otimes P_{i_j}.
\]
We concatenate $H_1,\cdots,H_t$ vertically to obtain the $\left(\sum_{j=1}^tm_j\right)\times \ell n$ matrix
\begin{equation}\label{PC-QC}
H=\begin{pmatrix}
H_1\\
\vdots\\
H_t
\end{pmatrix}
=\begin{pmatrix}
P_{i_1}&  \alpha^{i_1}P_{i_1} & \cdots& \alpha^{i_1(n-1)}P_{i_1}\\
\vdots& \vdots &  \ddots & \vdots\\
P_{i_t}&  \alpha^{i_t}P_{i_t} & \cdots& \alpha^{i_t(n-1)}P_{i_t}
\end{pmatrix}.
\end{equation}
It was shown in \cite{Semenov2012} that $\sum_{j=1}^tm_j=\sum_{j=0}^{\ell-1}\deg(g_{j,j}(x))$ and the matrix $H$ in \eqref{PC-QC} is a parity-check matrix of the code $\cc$. The \textit{eigencode} of a given eigenspace $V_{\alpha^{i_j}}$ is
\begin{equation}\label{AC-QC}
\mathbb{C}(V_{\alpha^{i_j}}) = \left\{\mathbf{v}\in\Ff_q^\ell : \mathbf{v} \cdot \mathbf{u}^{\top} = \sum_{a=0}^{\ell-1} v_a \, u_a=0\ \mbox{for all}\ \mathbf{u}\in V_{\alpha^{i_j}} \right\} \subseteq \Ff_q^\ell.
\end{equation}
In that same work, Semenov and Trifonov also proposed a BCH-like bound for quasi-cyclic codes. If $\cc$ is an $[\ell n,k,d]_q$ quasi-cyclic code with consecutive eigenvalue set $E=\{\alpha^{i_1},\cdots,\alpha^{i_t}\}$, then $d\geq\min\{t+1,d(\mathbb{C})\}$, with $\mathbb{C}$ being the eigencode of $\bigcap_{j=1}^tV_{\alpha^{i_j}}$. The BCH-like bound was generalized to an HT-like bound in \cite{Zeh2016}. Since these two bounds are defined based on the eigenvalues and eigencodes, they are called \emph{spectral bounds} for quasi-cyclic codes. A general case of spectral bounds was investigated by Ezerman \textit{et al.} in a recent work \cite{Ezerman2021}. The following theorem was proved as a main result.

\begin{thm}\label{thm21}{\rm \cite[Theorem 13]{Ezerman2021}}
Let $\cc$ be an $[\ell n,k,d]_q$ quasi-cyclic code of index $\ell$ with nonempty eigenvalue set $E\subseteq \Delta$. Let $\mathfrak{C}$ be the cyclic code of length $n$ over $\Ff_q$ with zero set $E$ and let $(L,d_L)\in D(\mathfrak{C})$ be a defining set bound given by Definition \ref{def1}. If $\mathbb{C}$ is the eigencode of $\bigcap_{\beta\in L}V_{\beta}$, then $d\geq \min\{d_L,d(\mathbb{C})\}$.
\end{thm}

By Theorem \ref{thm21}, the Roos bound \cite{Roos1983}, shift bound \cite{Lint1986}, and Betti-Sala bound \cite{Betti2006} can be used to estimate the minimum distances of quasi-cyclic codes.

\subsection{The Concatenated Structure of Quasi-Cyclic Codes}

This subsection describes the concatenated structure of a quasi-cyclic code with length $\ell n$ and index $\ell$. The respective proofs of the following assertions are provided in \cite{Guneri2013}.

Let the polynomial $x^n-1$ factor into irreducible polynomials in $\Ff_q[x]$ as $x^n-1=\prod_{i=1}^tf_i(x)$. Using the Chinese Remainder Theorem, we arrive at the ring isomorphism
\begin{equation}\label{s22eq1}
R=\Ff_q[x]/\langle x^n-1\rangle\cong\bigoplus_{i=1}^t\Ff_q[x]/\langle f_i(x)\rangle.
\end{equation}
We recall that $r$ is the smallest positive integer such that $q^r\equiv 1 \pmod n$ and $\alpha$ is a primitive $n^{\rm th}$ root of unity in $\Ff_{q^r}$. Since it is clear that the roots of each $f_i(x)$ are powers of $\alpha$, we let $v_i$ be the smallest nonnegative integer such that $f_i(\alpha^{v_i})=0$. Since the polynomial $f_i(x)$ is irreducible, each term $\Ff_q[x]/\langle f_i(x)\rangle$ in the direct summand in \eqref{s22eq1} is an extension field of $\Ff_q$, which we denote by $\mathbb{E}_i$. We rewrite \eqref{s22eq1} as
\begin{equation}\label{s22eq2}
R \cong\mathbb{E}_1\oplus \cdots \oplus \mathbb{E}_t \mbox{, sending }
a(x) \mapsto (a(\alpha^{v_1}),\cdots,a(\alpha^{v_t})).
\end{equation}
From \eqref{s22eq2}, we infer that $R^{\ell}\cong\mathbb{E}_1^{\ell}\oplus \cdots \oplus \mathbb{E}_t^{\ell}$. This representation allows for the identification of a quasi-cyclic code $\cc\subseteq R^{\ell}$ as an $(\mathbb{E}_1\oplus \cdots \oplus \mathbb{E}_t)$-submodule of $\mathbb{E}_1^{\ell}\oplus \cdots \oplus \mathbb{E}_t^{\ell}$. The code $\cc$ has the decomposition
\begin{equation}\label{Decomposition}
\cc\cong\cc_1\oplus \cdots \oplus \cc_t,
\end{equation}
where $\cc_i$ is a linear code of length $\ell$ over $\mathbb{E}_i$ and the codes $\cc_1,\cdots,\cc_t$ are the \emph{constituents} of $\cc$. Let the submodule $\cc$ of $R^\ell$ be formed by
\begin{equation}\label{Generators}
\{(a_{1,0}(x),\cdots,a_{1,\ell-1}(x)),\cdots,(a_{b,0}(x),\cdots,a_{b,\ell-1}(x))\}.
\end{equation}
Then we get the constituent code
\begin{equation}\label{constituent}
\cc_i={\rm Span}_{\mathbb{E}_i}\{(a_{j,0}(\alpha^{v_i}),\cdots,a_{j,\ell-1}(\alpha^{v_i})):1\leq j \leq b\}.
\end{equation}

Let $\mathcal{D}_i$ be an $[n,k_i,d_i]_q$ cyclic code with check polynomial $f_i(x)$ for each $i=1,\cdots,t$. The code $\mathcal{D}_i$ is minimal and isomorphic to the field $\mathbb{E}_i$. Let $\theta_i$ be the corresponding \textit{generating primitive idempotent} of $\mathcal{D}_i$. We identify $\langle\theta_i\rangle$ as the cyclic code $\mathcal{D}_i$. The isomorphism between $\mathcal{D}_i$ and $\mathbb{E}_i$ is given by

\begin{align}\label{eq:isomorphism}
\varphi_i &: \langle\theta_i\rangle \rightarrow \mathbb{E}_i \mbox{, sending } a(x)\mapsto a(\alpha^{v_i}) \mbox{ and} \notag \\
\psi_i &: \mathbb{E}_i \rightarrow \langle\theta_i\rangle \mbox{, sending } \beta \mapsto \sum_{j=0}^{n-1}a_jx^j,
\end{align}
with $a_j=\frac{1}{n}\Tr_{\mathbb{E}_i/\Ff_q}(\beta\alpha^{-jv_i})$. Based on the map $\psi_i$, the \emph{concatenated code} $\langle\theta_i\rangle\, \square \, \cc_i$ is
\begin{equation}\label{s22eq3}
\langle\theta_i\rangle\, \square \, \cc_i=\left\{\left(\psi_i(c_{i,0}),\cdots,\psi_i(c_{i,\ell-1})\right):(c_{i,0},\cdots,c_{i,\ell-1})\in\cc_i\right\}.
\end{equation}
Jensen in \cite{Jensen1985} showed that a quasi-cyclic code can be represented as the direct sum of concatenated codes. A lower bound on the minimum distance follows from this concatenated structure as well.

\begin{thm}\label{thmJensen}{\rm \cite{Jensen1985}}
Let $\cc$ be a quasi-cyclic codes of length $\ell n$ over $\Ff_q$. If $\cc_{i_1},\cdots,\cc_{i_t}$ are the nonzero constituents of $\cc$, then $\cc=\bigoplus_{z=1}^t \langle\theta_{i_z} \rangle \, \square \, \cc_{i_z}$. If $d(\cc_{i_1})\leq\cdots\leq d(\cc_{i_t})$, with $d(\cc_{i_z})$ being the minimum distance of the code $\cc_{i_z}$ for each $z=1,\cdots,t$, then the minimum distance $d$ of $\cc$ satisfies
\begin{equation}\label{Jensenbd}
d\geq d_J=\min_{1\leq z\leq t}\left\{d(\cc_{i_z}) \, d(\langle\theta_{i_1}\rangle\oplus\cdots\oplus\langle\theta_{i_z}\rangle)\right\}.
\end{equation}
\end{thm}

\section{Improved Spectral Bound for Quasi-Cyclic Codes}\label{S3}

In this section, we provide an improved spectral bound on the minimum distance of quasi-cyclic codes. This new bound outperforms the general spectral bound in Theorem \ref{thm21}.

Theorem \ref{thm21} shows that each defining set bound for cyclic codes is also applicable to quasi-cyclic codes. The minimum distance of quasi-cyclic codes with eigenvalue set $E$ is greater than or equal to $\min\{d_L,d(\mathbb{C})\}$, where $d_L$ and $d(\mathbb{C})$ are determined by the respective eigenvalue set and corresponding eigencode. To implement this spectral bound, it is practical to put the eigencode of $\bigcap_{\beta\in L}V_{\beta}$ to be $\{\mathbf{0}\}$, which implies that $d(\mathbb{C})=\infty$, as it is not easy to design an eigencode with a fixed minimum distance. Roth and Zeh in \cite{Roth2019} constructed an $\ell\times \ell$ upper-triangular polynomial matrix $\mathbf{G}^{\prime}(x)$ over $\Ff_q[x]$ with consecutive eigenvalue set $E=\{\alpha^{b},\alpha^{b+1},\cdots,\alpha^{b+\delta-2}\}$ and the eigencode $\mathbb{C}$ of $\bigcap_{\beta\in E}V_{\beta}$ being $\{\mathbf{0}\}$. Their approach requires the degree $[\Ff_q(\alpha):\Ff_q]$ of $\Ff_q(\alpha)$ over $\Ff_q$ to be $\geq \ell$ and $\deg(\det(\widetilde{\mathbf{G}}(x)))\geq \ell \, (\delta-1)$ when $[\Ff_q(\alpha):\Ff_q] = \ell$. The designed minimum distance of their quasi-cyclic codes is upper bounded by the co-index. To take advantage of the structure of quasi-cyclic codes in ensuring that the spectral bound is completely determined by the eigenvalue sets, we propose the following improved spectral bound for the said code family.

\begin{thm}\label{ISB}
Let $\cc$ be an $[\ell n,k,d]_q$ quasi-cyclic code of index $\ell$ with nonempty eigenvalue set $E$. Let $\mathfrak{C}$ be the cyclic code of length $n$ over $\Ff_q$ with zero set $E$ and let
\[
(L_1,d_{L_1}),\cdots,(L_s,d_{L_s})\in D(\mathfrak{C})
\]
be the defining set bounds given in Definition \ref{def1}. The sets $L_1,\cdots,L_s$ are not necessarily disjoint. Without loss of generality, we assume that $d_{L_1}\geq \cdots \geq d_{L_s}$. If $\mathbb{C}_i$ is the eigencode of $\bigcap_{\beta\in L_i}V_{\beta}$ for each $1\leq i\leq s$, then
\begin{equation}\label{NewSpec}
d\geq \min \left\{d_{L_1},d_{L_2} \, d(\mathbb{C}_1), d_{L_3} \, d\left(\bigcap_{i=1}^2\mathbb{C}_i\right), \cdots, d_{L_s} \, d\left(\bigcap_{i=1}^{s-1}\mathbb{C}_i\right), d\left(\bigcap_{i=1}^s\mathbb{C}_i\right)\right\}.
\end{equation}
\end{thm}
\begin{proof}
For each $1\leq i\leq s$, we label the elements of $L_i$ as $\beta_1^{(i)},\cdots,\beta_{j_i}^{(i)}$ and let $P_i=\left(\mathbf{v}_1^{(i)},\cdots,\mathbf{v}_{m_i}^{(i)}\right)^{\top}$ be an $m_i\times \ell$ matrix whose rows are formed by a basis of $\bigcap_{\beta\in L_i}V_{\beta}$. Let
\[
T_i=\begin{pmatrix}
1&  \beta_1^{(i)} & \cdots& \left(\beta_1^{(i)}\right)^{n-1}\\
\vdots& \vdots &  \ddots & \vdots\\
1&  \beta_{j_i}^{(i)} & \cdots& \left(\beta_{j_i}^{(i)}\right)^{n-1}
\end{pmatrix}
\]
be a $j_i \times n$ matrix. We can verify that $T_i$ is a parity-check matrix of some cyclic code $\mathcal{D}_i$ with zero set $L_i$. The minimum distance of $\mathcal{D}_i$ is greater than or equal to $d_{L_i}$. By \eqref{PC-QC}, the matrix
\[
\widetilde{H}=\begin{pmatrix}
T_1\otimes P_1\\
\vdots\\
T_s\otimes P_s
\end{pmatrix}
\]
is a submatrix of a parity-check matrix of $\cc$ defined in \eqref{PC-QC}. The code $\widetilde{\cc}$ over $\Ff_q$ with parity-check matrix $\widetilde{H}$ contains the code $\cc$, which implies that $d\geq d(\widetilde{\cc})$.

Let $\mathbf{c}=(\mathbf{c}_0,\cdots,\mathbf{c}_{n-1})$, with $\mathbf{c}_{i}=(c_{i,0},\cdots,c_{i,\ell-1})$ being a nonzero codeword of $\widetilde{\cc}$. The fact that $\widetilde{H}\cdot \mathbf{c}^{\top}=\mathbf{0}$ implies that $T_i\cdot\left(P_i\mathbf{c}_0^{\top},\cdots,P_i\mathbf{c}_{n-1}^{\top}\right)^{\top}=O$ with $O$ being the $j_i \times m_i$ zero matrix for each $1\leq i\leq s$. Since $P_i = \left(\mathbf{v}_1^{(i)},\cdots,\mathbf{v}_{m_i}^{(i)}\right)^{\top}$, we obtain
\[
T_i\cdot \begin{pmatrix}
\mathbf{v}_1^{(i)}\mathbf{c}_0^{\top} & \cdots & \mathbf{v}_{m_i}^{(i)}\mathbf{c}_0^{\top} \\
\vdots&  \ddots & \vdots\\
\mathbf{v}_1^{(i)}\mathbf{c}_{n-1}^{\top} & \cdots & \mathbf{v}_{m_i}^{(i)}\mathbf{c}_{n-1}^{\top}\\
\end{pmatrix}=O.
\]

We now divide the justification into four cases. First, we consider the case when the matrix $\left(P_i\mathbf{c}_0^{\top},\cdots,P_i\mathbf{c}_{n-1}^{\top}\right)$ is not the zero matrix for each $i$. Second, we let $\left(P_i\mathbf{c}_0^{\top},\cdots,P_i\mathbf{c}_{n-1}^{\top}\right)$ be the zero matrix for some $i$. Finally, we take $\left(P_i\mathbf{c}_0^{\top},\cdots,P_i\mathbf{c}_{n-1}^{\top}\right)$ to be the zero matrix for all $i$.
\begin{enumerate}[wide, itemsep=0pt, leftmargin =0pt, widest={{\bf Case $2$}}]
\item[{\bf Case $1$}:] The matrix $\left(P_i\mathbf{c}_0^{\top},\cdots,P_i\mathbf{c}_{n-1}^{\top}\right)$ is not the zero matrix for each $1\leq i\leq s$. Hence, there exists some nonzero column $\left(\mathbf{v}_t^{(i)}\mathbf{c}_0^{\top},\cdots,\mathbf{v}_t^{(i)}\mathbf{c}_{n-1}^{\top}\right)^{\top}$ of $\left(P_i\mathbf{c}_0^{\top},\cdots,P_i\mathbf{c}_{n-1}^{\top}\right)^{\top}$. Since
\[
T_i\cdot\left(\mathbf{v}_t^{(i)}\mathbf{c}_0^{\top},\cdots,\mathbf{v}_t^{(i)}\mathbf{c}_{n-1}^{\top}\right)^{\top}=\mathbf{0}^{\top},
\]
the vector $\left(\mathbf{v}_t^{(i)}\mathbf{c}_0^{\top},\cdots,\mathbf{v}_t^{(i)}\mathbf{c}_{n-1}^{\top}\right)$ has at least $d_{L_i}$ nonzero entries. If $\mathbf{v}_t^{(i)}\mathbf{c}_a^{\top}\neq 0$ for some $a$, then there is at least one nonzero coordinate in $\mathbf{c}_{a}=(c_{a,0},\cdots,c_{a,\ell-1})$. Hence, the Hamming weight of $\mathbf{c}$ is at least $d_{L_i}$. Since $d_{L_1}\geq\cdots\geq d_{L_s}$ and the code $\widetilde{\cc}_i$ with parity-check matrix $T_i\otimes P_i$ contains $\widetilde{\cc}$ for each $1\leq i\leq s$, the Hamming weight of $\mathbf{c}$ is at least $d_{L_1}$.
\item[{\bf Case $2$}:] Let $u$ be a positive integer such that $1\leq u\leq s-2$. For each subset $X$ of $[s]$ with cardinality $u$, let the matrix $\left(P_i\mathbf{c}_0^{\top},\cdots,P_i\mathbf{c}_{n-1}^{\top}\right)$ be the zero matrix for each $i\in X$, whereas the matrix $\left(P_i\mathbf{c}_0^{\top},\cdots,P_i\mathbf{c}_{n-1}^{\top}\right)$ is not the zero matrix for each $i\in([s]\setminus X)$. We note that $\mathbb{C}_i$ is the eigencode of $\bigcap_{\beta\in L_i}V_{\beta}$ for each $1\leq i\leq s$. By the first assumption, $\mathbf{c}_0,\cdots,\mathbf{c}_{n-1}$ are codewords of the code  $\bigcap_{i\in X}\mathbb{C}_i$. Using a similar method as in Case 1, the second assumption implies the existence of a vector $\left(\mathbf{v}_t^{(j)}\mathbf{c}_0^{\top},\cdots,\mathbf{v}_t^{(j)}\mathbf{c}_{n-1}^{\top}\right)$ having at least $d_{L_j}$ nonzero coordinates for $j\in([s]\setminus X)$. If $\mathbf{v}_t^{(j)}\mathbf{c}_a^{\top}\neq 0$ for some $a$, then there are at least $d\left(\bigcap_{i\in X}\mathbb{C}_i\right)$ nonzero coordinates in $\mathbf{c}_{a}=(c_{a,0},\cdots,c_{a,\ell-1})$. Thus, an argument similar to the one used in Case 1 shows that the Hamming weight of $\mathbf{c}$ is at least
\[
d_{u}=\max\left\{d_{L_j} d\left(\bigcap_{i\in X}\mathbb{C}_i\right):X\subset [s], |X|=u,j\in ([s]\setminus X)\right\}.
\]
\item[{\bf Case $3$}:]For each $X \subseteq [s]$ with cardinality $s-1$, let $\left(P_i\mathbf{c}_0^{\top},\cdots,P_i\mathbf{c}_{n-1}^{\top}\right)$ be the zero matrix for each $i\in X$, whereas the matrix $\left(P_i\mathbf{c}_0^{\top},\cdots,P_i\mathbf{c}_{n-1}^{\top}\right)$ is not the zero matrix for each $i\in([s]\setminus X)$. Using the same argument as in Case 2, we infer that the Hamming weight of $\mathbf{c}$ belongs to the set
\[
\left\{d_{L_j} d\left(\bigcap_{i\in ([s]\setminus \{j\})}\mathbb{C}_i\right):j\in[s]\right\}.
\]
This implies that its Hamming weight is at least
\[
d_{s-1}=\min\left\{d_{L_j} d\left(\bigcap_{i\in ([s]\setminus \{j\})}\mathbb{C}_i\right):j\in[s]\right\}.
\]
Hence, each row of $\widetilde{H}$ is used to estimate the Hamming weight of $\mathbf{c}$. Unlike in Case 1 and Case 2, the codeword $\mathbf{c}$ is no longer in the intersection of some subcodes of $\widetilde{\cc}$. This explains why $d_{s-1}$ takes a different form than those of $d_{1},\cdots,d_{s-2}$.

\item[{\bf Case $4$}:] The matrix $\left(P_i\mathbf{c}_0^{\top},\cdots,P_i\mathbf{c}_{n-1}^{\top}\right)$ is the zero matrix for each $1\leq i\leq s$. Hence, $ \mathbf{c}_0,\cdots,\mathbf{c}_{n-1}$ are codewords of the code $\bigcap_{i=1}^{s}\mathbb{C}_i$. Since the codeword $\mathbf{c}=(\mathbf{c}_0,\cdots,\mathbf{c}_{n-1})$ is nonzero, the Hamming weight of $\mathbf{c}$ is at least $d\left(\bigcap_{i=1}^{s}\mathbb{C}_i\right)$.
\end{enumerate}

Summarizing the above cases, we conclude that
\[
d\geq \min\left\{d_{L_1}, d_{1},\cdots,d_{s-1},d\left(\bigcap_{i=1}^s\mathbb{C}_i\right)\right\}.
\]
It is then straightforward to verify that
$\displaystyle{d_{u} \geq d_{L_{u+1}} \, d\left(\bigcap_{i=1}^{u}\mathbb{C}_i\right)}$ and $\displaystyle{d_{L_{u+1}} \, d\left(\bigcap_{i\in([s]\setminus\{u+1\})}\mathbb{C}_i\right) \geq d_{L_{u+1}} \, d\left(\bigcap_{i=1}^{u}\mathbb{C}_i\right)}$ for each $u\in[s-2]$. Thus, for simplicity, we henceforth adopt the improved spectral bound in \eqref{NewSpec}.
\end{proof}

\begin{remark}\label{gen-matrix-remark}
In Theorem \ref{ISB}, we consider a quasi-cyclic code $\cc$ with reduced generator polynomial matrix $\widetilde{\mathbf{G}}(x)$. In fact, Theorem \ref{ISB} also holds for the generator polynomial matrix $\mathbf{G}^{\prime}(x)$ since the set of eigenvalues remains the same.
\end{remark}

The lower bound in Theorem \ref{ISB} may look complicated at first sight. The bound simplifies to $\min\left\{d_{L_1}, d_{L_2} \, d(\mathbb{C}_1), d \left(\mathbb{C}_1 \bigcap  \mathbb{C}_2 \right)\right\}$ when we fix $s=2$. Such a restriction, as we shall see in Table \ref{table:simulation}, is often sufficient for comparative purposes. Given a quasi-cyclic code of index $\ell$ and generator polynomial matrix $\widetilde{\mathbf{G}}(x)$, we can select arbitrary subsets $L_1,\cdots,L_s$ of the eigenvalue set. In order to compute or estimate the minimum distances of the codes $\bigcap_{i=1}^j \mathbb{C}_i$, for each $j\in[s]$, we can select some suitable subsets $L_1,\cdots,L_s$ and directly obtain the values of $d_{L_1},\cdots,d_{L_s}$ from the defining set bounds of the corresponding cyclic codes. We know from the structural analysis in \cite{Ezerman2021} that the eigencode $\mathbb{C}_i$ of length $\ell$ over $\Ff_q$ can be built by evaluating $\widetilde{\mathbf{G}}(x)$ at the elements of $L_i$. The following corollary states a lower bound for quasi-cyclic codes of index $2$ by calculating the minimum distances of the codes $\bigcap_{i=1}^j \mathbb{C}_i$ for each $j\in[s]$.

\begin{cor}\label{cor-3-1}
Let $\cc$ be a $[2n,k,d]_q$ quasi-cyclic code of index $2$ with polynomial generator matrix
\[
\widetilde{\mathbf{G}}(x) =
\begin{pmatrix}
g_1(x)  & h(x)  \\
0 &  g_2(x) 
\end{pmatrix}.
\]
Let $E_1$ and $E_2$ be the respective collections of all the roots of $g_1(x)$ and $g_2(x)$ in some extension fields of $\Ff_q$. Let $E=E_1\bigcup E_2$, $L_1\subseteq E_1$, and $L_2\subseteq E_2$. Let $\mathfrak{C}$ be the cyclic code of length $n$ over $\Ff_q$ with zero set $E$ and let $(L_1,d_{L_1}),(L_2,d_{L_2})\in D(\mathfrak{C})$ be the defining set bounds as in Definition \ref{def1}. If there exist $a\in L_1$ and $b\in L_2$ such that $g_2(a)\neq 0$, $g_1(b)\neq 0$, and $h(b)\neq 0$, then $d\geq \min\{d_{L_1},d_{L_2}\}$.
\end{cor}
\begin{proof}
Let $\mathbb{C}_i$ be the eigencode of $\bigcap_{\beta\in L_i}V_{\beta}$ for $i \in \{1,2\}$. For each $\beta\in L_1$ and for each $(u_1,u_2)\in \bigcap_{\beta\in L_1}V_{\beta}$, we have 
\[
\begin{pmatrix}
g_1(\beta)  & h(\beta)  \\
0 &  g_2(\beta) 
\end{pmatrix}
\begin{pmatrix}
u_1\\
u_2
\end{pmatrix}
=
\begin{pmatrix}
0\\
0
\end{pmatrix}.
\]
Since $L_1\subseteq E_1$, we get $g_1(\beta)=0$, which implies that $u_2 \, h(\beta)=u_2 \, g_2(\beta)=0$ for each $\beta\in L_1$. The existence of an $a \in L_1$ such that $g_2(a) \neq 0$ gives us $u_2=0$. Hence, letting $r$ be the smallest positive integer such that $q^r\equiv 1 \pmod n$,
\[
\bigcap_{\beta\in L_1}V_{\beta}=\left\{(u_1,0):u_1\in\Ff_{q^r}\right\}.
\]
It is immediate to confirm that $\mathbb{C}_1=\{(0,u^{\prime}):u^{\prime}\in\Ff_{q}\}$. In a similar manner, we can infer that $\mathbb{C}_2=\{(u^{\prime},0):u^{\prime}\in\Ff_{q}\}$. Thus, $\mathbb{C}_1 \bigcap\mathbb{C}_2=\{(0,0)\}$ and, by Theorem \ref{ISB}, $d\geq \min\{d_{L_1},d_{L_2}\}$.
\end{proof}

The lower bound for quasi-cyclic codes of index $2$ in Corollary \ref{cor-3-1} relies on some conditions. When the conditions are not all met but the index remains $2$, a similar argument as in the proof of Corollary \ref{cor-3-1} can be devised to establish a lower bound on the minimum distance. Later, in Theorems \ref{QC-Design} and \ref{LRC-C3}, we will determine the exact values of the minimum distances of the codes $\bigcap_{i=1}^j \mathbb{C}_i$ in two infinite families of quasi-cyclic codes with large indices.

Let $\cc$ be an $[\ell n,k,d]_q$ quasi-cyclic code of index $\ell$ with nonempty eigenvalue set $E$. Let $\mathfrak{C}$ be the cyclic code of length $n$ over $\Ff_q$ with zero set $E$ and let
\[
(L_1,d_{L_1}),\cdots,(L_s,d_{L_s})\in D(\mathfrak{C})
\]
be the defining set bounds for $\mathfrak{C}$. We, henceforth, denote the estimate of the improved spectral bound by
\[
d_{Spec}(D(\mathfrak{C}); L_1,\cdots,L_s) := \min\left\{d_{L_1}, d_{L_2} \, d(\mathbb{C}_1), \cdots, d_{L_s} \, d\left(\bigcap_{i=1}^{s-1}\mathbb{C}_i\right), d\left(\bigcap_{i=1}^s\mathbb{C}_i\right)\right\},
\]
and we set
\begin{equation}\label{optbound}
d_{Spec}(D(\mathfrak{C}),s):=\max_{\substack{(L_i,d_i)\in D(\mathfrak{C}) \\ 1\leq i \leq s}} \{d_{Spec}(D(\mathfrak{C}); L_1,cldots,L_s)\}.
\end{equation}
The improved spectral bound $d_{Spec}(D(\mathfrak{C}); L_1,\cdots,L_s)$ considers $s$ out of $|D(\mathfrak{C})|$ defining set bounds and $d_{Spec}(D(\mathfrak{C}),s)$ maximizes over all $ |D(\mathfrak{C})|^{s}$ possible outcomes.

\begin{remark}\label{remark1}
Theorem \ref{ISB} is a generalization of Theorem \ref{thm21} as the latter follows when $s=1$ in Theorem \ref{ISB}. Whenever $s>1$, since $d \left(\bigcap_{i=1}^s \mathbb{C}_i\right) \geq d\left(\bigcap_{i=1}^{s-1} \mathbb{C}_i\right) \geq \cdots \geq d(\mathbb{C}_1 \cap \mathbb{C}_2) \geq d(\mathbb{C}_1)$ and $d_{L_s}\geq 1$, we obtain
\[
\min\left\{d_{L_1}, d_{L_2} \, d(\mathbb{C}_1),\cdots, d_{L_s} \, d\left(\bigcap_{i=1}^{s-1}\mathbb{C}_i\right), d \left(\bigcap_{i=1}^s \mathbb{C}_i \right)\right\}\geq \min\{d_{L_1}, d(\mathbb{C}_1)\}.
\]

Let $(L_i,d_{L_i})=(L_j,d_{L_j})$, for some $1\leq i<j\leq s$. Without loss of generality, we rearrange the list and assume that $j=i+1$. Since $d_{L_i}=d_{L_{i+1}}$ and $\mathbb{C}_i=\mathbb{C}_{i+1}$, we have $d_{L_i} \, d\left(\bigcap_{k=1}^{i}\mathbb{C}_k\right)=d_{L_{i+1}} \, d\left(\bigcap_{k=1}^{i+1}\mathbb{C}_k\right)$. Hence, considering an identical tuples $(L_i,d_{L_i})=(L_j,d_{L_j})$ among the $s$ terms is equivalent to considering the bound given in Theorem \ref{ISB} with $s-1$ terms. Therefore, the bound with $s$ terms, optimized as in \eqref{optbound}, already includes the respective estimates produced with $s-1$, $s-2$, $\cdots$, $2$ terms and a single term. We note that, in this repetition-free consideration, $s \leq |D(\mathfrak{C})|$.

To actually compute the bound in Theorem \ref{ISB}, we set $\bigcap_{i=1}^s\mathbb{C}_i=\{\mathbf{0}\}$ when $\ell < n$, as it makes it easier to carry out the computation, instead of when $\mathbb{C}_1=\{\mathbf{0}\}$. Under the preferred setup, the improved spectral bound can be greater than $\ell$. If $\bigcap_{i=1}^s\mathbb{C}_i=\{\mathbf{0}\}$, then the improved spectral bound in Theorem \ref{ISB} is determined by the eigenvalue set. The spectral bound in Theorem \ref{thm21} cannot exceed $\max\{n,\ell\}$. If $L_1$ is the set of roots of $x^n-1$ and $\bigcap_{i=1}^s\mathbb{C}_i=\{\mathbf{0}\}$, then the improved spectral bound in Theorem \ref{ISB} becomes greater than $\max\{n,\ell\}$ under certain conditions.
\end{remark}

\begin{example}\label{example1}
The $[9,2,6]_2$ quasi-cyclic code, with $(n,\ell)=(3,3)$, generated by $(a_{1,0}(x),\cdots,a_{1,2}(x))=(x^2 + x, x + 1, x^2 + 1)$ (cf. (\ref{Generators})) is optimal (see \cite{Grassl}). Its upper-triangular generator matrix is
\[
\mathbf{G}^{\prime}(x) =
\begin{pmatrix}
x + 1  & x^2 + 1 & x^2 +x \\
0 &  x^3 + 1 & 0\\
0 & 0 & x^3+1
\end{pmatrix},
\]
with $\det(\widetilde{\mathbf{G}}(x))=(x+1)(x^3+1)^2$ and, hence, $E=\Delta$. Over $\F_2$, the scalar generator matrix is
\[
\setcounter{MaxMatrixCols}{20}
\begin{pmatrix}
1 & 0 & 1 & 0 & 1 & 1 & 1 & 1 & 0 \\
0 & 1 & 1 & 1 & 1 & 0 & 1 & 0 & 1
\end{pmatrix}.
\]
The improved spectral bound yields $d_{Spec}(\widehat{D}(\mathfrak{C}),2)=6$, which is sharp, whereas, with $n=\ell=3$, the spectral bound in Theorem \ref{thm21} cannot exceed $3$. In this example, it yields $3$ as the estimate.\hfill $\blacksquare$
\end{example}

The formulation of the lower bound in Theorem \ref{ISB} is flexible due to the arbitrariness of $s$ and $L_1,\cdots,L_s$. We can select suitable formulations of the lower bound relative to different requirements. To achieve a tighter bound, we choose the $d_{Spec}(D(\mathfrak{C}),s)$ defined in \eqref{optbound}. For computational convenience, we adopt a small $s$. We now give a simple formulation of the lower bound for quasi-cyclic codes under a set of conditions.

\begin{cor}\label{cor-lrc-1}
Let $g(x)\in\Ff_q[x]$ be a divisor of $x^n-1$ and let $R_{g(x)}$ be the collection of all roots of $g(x)$ in some extension field of $\Ff_q$. Let $\cc$ be an $[\ell n,k,d]_q$ quasi-cyclic code of index $\ell$ with nonempty eigenvalue set $E$ and generator polynomial matrix $\widetilde{\mathbf{G}}(x)=g(x) \, \overline{\mathbf{G}}(x)$. Let $\mathfrak{C}$ be the cyclic code of length $n$ over $\Ff_q$ with zero set $E$ and let $(L_1,d_{L_1}),(L_2,d_{L_2})\in D(\mathfrak{C})$ be the defining set bounds in Definition \ref{def1}. If $\mathbb{C}_1$ is the eigencode of $\bigcap_{\beta\in L_1}V_{\beta}$, then the following assertions hold.
\begin{enumerate}
\item If $L_1,L_2\subseteq R_{g(x)}$, then $d\geq d_{L_1}$. 
\item If $L_1$ is not a subset of $R_{g(x)}$ and $L_2\subseteq R_{g(x)}$, then $d\geq \min \left\{d_{L_1},d_{L_2} \, d(\mathbb{C}_1)\right\}$.
\end{enumerate}
\end{cor}
\begin{proof}
By definitions of eigenspaces and eigencodes, for $i \in \{1,2\}$, we have $\mathbb{C}_i=\{\mathbf{0}\}$ whenever $L_i\subseteq R_{g(x)}$. The respective lower bounds follow from Theorem \ref{ISB}.
\end{proof}

The quasi-cyclic codes defined in Corollary \ref{cor-lrc-1} are particularly useful in large-scale distributed storage systems. We will show in Section \ref{S6} that these quasi-cyclic codes are locally repairable codes. Theorem \ref{LRC-C3} will give us the exact values of $d(\mathbb{C}_1)$ for codes in an infinite family of locally repairable codes with optimal parameters after a discussion on the construction approach for such a family.

\section{Comparison Results}\label{S4}

Based on the simulation results in \cite[Section V]{Ezerman2021}, the Jensen bound, which we reproduce here as Theorem \ref{thmJensen}, outperforms the spectral bound in Theorem \ref{thm21}. That spectral bound, in turn, is sharper than the Lally bound in \cite{Lally2003}. This section considers the performance of the improved spectral bound in Theorem \ref{ISB} against the Jensen bound (denoted by $d_J$) and the spectral bound (denoted by $d_S$). Although Remark \ref{remark1} says that the improved spectral bound is tighter than the spectral bound, we add the spectral bound to the comparison to underline the qualitative improvement.

Recall that, given a quasi-cyclic code with nonempty eigenvalue set $E\subseteq \Delta$, we let $\mathfrak{C}$ be the cyclic code of length $n$ over $\Ff_q$ with zero set $E$. In the examples provided below, we take a subset $\widehat{D}(\mathfrak{C})$ of the defining set bounds $D(\mathfrak{C})$ consisting of four types of bounds, namely the BCH, HT, and Roos bounds, as well as $(P,d(\mathfrak{C}_P)) \in D(\mathfrak{C})$, where $P$ is a subset of the eigenvalue set of the given quasi-cyclic code and $\mathfrak{C}_P$ is the corresponding cyclic code with zero set $P$, which contains $\mathfrak{C}$ as a subcode. In the search routines, we use the \texttt{MAGMA} functions for the Roos, HT, and BCH bounds from \cite{Piva2014}. The restricted estimate of the improved spectral bound $d_{Spec}(\widehat{D}(\mathfrak{C}),s)$ is calculated for $s \in \{2,3,4\}$.

\begin{example}
The $[12,3,4]_2$ quasi-cyclic code with $(n,\ell)=(3,4)$ generated, in the notation of (\ref{Generators}), by 
\[
(a_{1,0}(x),\cdots,a_{1,3}(x))=(x^2 + 1, \, x^2 + x + 1, \, x^2 + 1, \, x^2 + x + 1)
\]
has an upper-triangular generator matrix
\[
\mathbf{G}^{\prime}(x) =
\begin{pmatrix}
x + 1  & 0 & x+1 & 0\\
0 &  x^2 + x + 1 & 0 & x^2 + x + 1\\
0 & 0 & x^3+1 & 0\\
0 & 0 & 0 & x^3+1
\end{pmatrix}.
\]
Since $\det(\widetilde{\mathbf{G}}(x))=(x^3+1)^3$, we know that $E=\Delta$. Over $\F_2$, the scalar generator matrix is
\[
\setcounter{MaxMatrixCols}{20}
\begin{pmatrix}
1 & 0 & 1 & 0 & 0 & 0 & 1 & 0 & 1 & 0 & 0 & 0 \\
0 & 1 & 1 & 0 & 0 & 0 & 0 & 1 & 1 & 0 & 0 & 0 \\
0 & 0 & 0 & 1 & 1 & 1 & 0 & 0 & 0 & 1 & 1 & 1
\end{pmatrix}.
\]
The improved spectral bound in (\ref{NewSpec}) yields $d_{Spec}(\widehat{D}(\mathfrak{C}),3)=4$, whereas the Jensen bound in (\ref{Jensenbd}) and the spectral bound in Theorem \ref{thm21} give the same estimate $d_J=d_S=2$.\hfill $\blacksquare$
\end{example}

\begin{example}
The $[8,6,2]_3$ quasi-cyclic code with $(n,\ell)=(4,2)$ generated by
\[
(a_{1,0}(x), a_{1,1}(x))=(x^3 + 2x^2 + x + 2, \, x^2 + 2x + 1)
\mbox{ and } (a_{2,0}(x), a_{2,1}(x))=(2x^2 + x, \, x^3 + x^2 + x + 1)
\]
in the notation of (\ref{Generators}) is optimal (see \cite{Grassl}) and has an upper-triangular generator matrix
\[
\widetilde{\mathbf{G}}(x) =
\begin{pmatrix}
x + 2  & 0 \\
0 &  x + 1 \\
\end{pmatrix}.
\]
Since $\det(\widetilde{\mathbf{G}}(x)) =x^2+2 =(x+2) \, (x+1)$, we infer that $E=\{1,2\}$. Over $\F_3$, the scalar generator matrix is
\[
\setcounter{MaxMatrixCols}{20}
\begin{pmatrix}
1 & 0 & 0 & 2 & 0 & 0 & 0 & 0 \\
0 & 1 & 0 & 2 & 0 & 0 & 0 & 0 \\
0 & 0 & 1 & 2 & 0 & 0 & 0 & 0 \\
0 & 0 & 0 & 0 & 1 & 0 & 0 & 1 \\
0 & 0 & 0 & 0 & 0 & 1 & 0 & 2 \\
0 & 0 & 0 & 0 & 0 & 0 & 1 & 1 \\
\end{pmatrix}. \vspace{10pt}
\]
The improved spectral bound in (\ref{NewSpec}) yields $d_{Spec}(\widehat{D}(\mathfrak{C}),2)=2$, whereas the Jensen bound in (\ref{Jensenbd}) and the spectral bound in Theorem \ref{thm21} provide the same estimate $d_J=d_S=1$. \hfill $\blacksquare$
\end{example}

To determine the improved spectral bound with $s$ terms, given a quasi-cyclic code of length $n\ell$ and eigenvalue set $E$, one has to consider $(2^{|E|}-1)^s$ possible $s$-tuples of nonempty subsets $L_1, \cdots, L_s \subseteq E$. Furthermore, if $k$ types of defining set bounds are considered in $D(\mathfrak{C})$, \textit{e.g.}, the BCH, HT, Roos, \textit{etc.}, then $k^s$ possible $s$-tuples of defining set bounds have to be taken into consideration per choice of $L_1, \cdots, L_s \subseteq E$. In total, there are $k^s \left(2^{|E|}-1\right)^s$ comparisons to be done for each randomly generated quasi-cyclic code. In Table \ref{table:simulation} and in the examples above, we have taken $k=4$ and choose $s \in \{2,3,4\}$ to keep the number of comparisons practical.

Table \ref{table:simulation} presents the comparison of the overall performances of the improved spectral bound, the Jensen bound, and the spectral bound in the binary and ternary cases. For $q=2$, we constructed random binary quasi-cyclic codes on each input tuple $(n,\ell,r)$, with $n = 3$ fixed, to keep the size of the eigenvalue set $|E|$ manageable, and use the ranges $2 \leq \ell \leq 4$ and $1 \leq r \leq \ell$. Once $(n,\ell,r)$ was determined, an array $\mathcal{A}$ of $r \ell$ generator polynomials was randomly built. The {\tt QuasiCyclicCode} function of \texttt{MAGMA} on input $(n \ell, \mathcal{A}, r)$ produced the corresponding $r$-generator quasi-cyclic code $C$. If $C$ was nontrivial, then we computed its minimum distance $d(C)$ and the estimates given by the three bounds $d_{Spec}, d_J$ and $d_{S}$, where $d_{Spec} :=d_{Spec}(\widehat{D}(\mathfrak{C}),s)$, with $s=2,3$ or $4$ and $\widehat{D}(\mathfrak{C})$ consisted of the $k=4$ types of defining set bounds, and $d_{S}$ was maximized over all possible choices $(L,d_L)\in D(\mathfrak{C})$, as had been done in \cite{Ezerman2021}. Running the complete simulation several times revealed that the number of nontrivial codes constructed and their performance relative to the bounds vary negligibly for fixed $q$ and $s$. The entries in the table were simply taken from the last simulation run for each $(q,s)$ tuple. One can, of course, opt to include the counts from many runs. Doing so, however, would not meaningfully change the percentage. We recorded the respective numbers of occasions when each bound was either sharp or best-performing, with double counting allowed.

An identical routine was carried out for $q=3$, with $n=4$ fixed, $2 \leq \ell \leq 4$, and $1 \leq r \leq \ell$. We constructed random ternary quasi-cyclic codes on each input tuple $(n,\ell,r)$ and recorded the performance of the estimates given by the three bounds $d_{Spec}, d_J$ and $d_{S}$, where $d_{Spec} :=d_{Spec}(\widehat{D}(\mathfrak{C}),s)$, again with $s \in \{2,3,4\}$ and $d_{S}$ was maximized over all possible choices $(L,d_L)\in D(\mathfrak{C})$. The table confirms that the improved spectral bound performs better than the Jensen bound and the spectral bound in general, whereas the Jensen bound was listed as the best performing bound in \cite{Ezerman2021}.

\begin{table*}[ht!]
\caption{The outcomes of performance comparison of the minimum distance bounds in the binary and ternary cases. We list the number of nontrivial instances, for a search round, with double counting allowed, when a specified bound reached the actual minimum distance, \textit{i.e.}, it was sharp, and when it was greater than or equal to the other bound, \textit{i.e.}, it was best-performing.}
\label{table:simulation}
\renewcommand{\arraystretch}{1.2}
\centering
\begin{tabular}{c|c|c|c|c|c}
	\hline
	$q=2$ &  $d_{Spec}, s=2$ &  $d_{Spec}, s=3$ &  $d_{Spec}, s=4$ & $d_J$ & $d_{S}$    \\
	\hline
	sharp & $154$ & $156$ & $156$ & $144$ & $118$  \\
	best-performing & $187$ & $189$ & $189$ & $179$ & $140$ \\
	\hline
	  \multicolumn{6}{c}{ \# nontrivial $C$ : $207$ }  \\
	\hline
	\hline
	$q=3$ &  $d_{Spec}, s=2$ &  $d_{Spec}, s=3$ &  $d_{Spec}, s=4$ & $d_J$ & $d_{S}$    \\
	\hline
	sharp & 40 & 41 & 41 & 35 & 31  \\
	best-performing & 66 & 69 & 69 & 69 & 51 \\
	\hline
	  \multicolumn{6}{c}{ \# nontrivial $C$ : 82 }  \\
	\hline
\end{tabular}
\end{table*}

\begin{comment}
%%% Images
\begin{figure*}[h!]
\centering
\begin{tabular}{cc}
	\includegraphics[width=0.45\linewidth]{new-binary_s_2.png} &
	\includegraphics[width=0.45\linewidth]{new-binary_s_3.png} \\
	(a) Binary QC codes with $s=2$ & (b) Binary QC codes with $s=3$\\
	\includegraphics[width=0.45\linewidth]{new-binary_s_4.png} &
	\includegraphics[width=0.45\linewidth]{new-ternary_s_2.png} \\
	(c) Binary QC codes with $s=4$ & (d) Ternary QC codes with $s=2$\\
	\includegraphics[width=0.45\linewidth]{new-ternary_s_3.png} &
	\includegraphics[width=0.45\linewidth]{new-ternary_s_4.png} \\
	(e) Ternary QC codes with $s=3$ & (f) Ternary QC codes with $s=4$\\
\end{tabular}
\caption{Performance comparison of the improved spectral bound ({\tt New Spectral}), the Jensen bound, and the previous spectral bound ({\tt Old Spectral}) based on the computational evidence in Table \ref{table:simulation}.}
\label{fig:one}
\end{figure*}
\end{comment}

\section{Quasi-Cyclic Codes with Designed Minimum Distances}\label{S5}

We now give a construction of quasi-cyclic codes with designed minimum distances from the improved spectral bound in Theorem \ref{ISB}.

Let $n$ and $\ell$ be positive integers with $\gcd(n,q)=1$ and $\ell<q$. Let $\mathfrak{g}(x)\in\Ff_q[x]$ be a divisor of $x^n-1$ with factorization
\[
\mathfrak{g}(x)=\prod_{i=1}^\ell \, \prod_{j=1}^{m_i}\mathfrak{g}_{i,j}(x),
\]
where $\mathfrak{g}_{i,j}(x)\in\Ff_q[x]$ is an irreducible polynomial for each $i\in[\ell]$ and $j\in[m_i]$. We recall that $r$ is the smallest positive integer satisfying $q^r\equiv 1 \pmod n$. Let $E_{i,j}=\{\beta\in\Ff_{q^r}: \mathfrak{g}_{i,j}(\beta)=0\}$ and $E=\bigcup_{i=1}^\ell \bigcup_{j=1}^{m_i} E_{i,j}$. Let $\gamma_1,\cdots,\gamma_{\ell-1}$ be distinct elements of $\Ff_{q}\setminus\{0,1\}$. For each $u\in[\ell-1]$, we define the polynomial $h_u(x)$ over $\Ff_q$ as
\begin{equation}\label{Polynomial-1}
h_u(x) = \sum_{i=1}^\ell \, \sum_{j=1}^{m_i} a_{i,j,u}(x) \, y_{i,j}(x) \mbox{, where } y_{i,j}(x)=\frac{\mathfrak{g}(x)}{\mathfrak{g}_{i,j}(x)}  \mbox{ and }
a_{i,j,u}(\beta) = -\gamma_u^i \left(y_{i,j}(\beta)\right)^{-1},
\end{equation}
for each $\beta\in E_{i,j}$. We proceed to construct the following $\ell\times \ell$ upper triangular matrix over $\Ff_q[x]$
\begin{equation}\label{Design1}
\mathbf{G}^{\prime}(x)=\begin{pmatrix}
1&0&\cdots&0&h_{1}(x)\\
0&1&\cdots&0&h_{2}(x)\\
\vdots& \vdots &  \ddots & \vdots & \vdots\\
0&0&\cdots&1&h_{\ell-1}(x)\\
0&0&\cdots&0&\mathfrak{g}(x)
\end{pmatrix}
\end{equation}
to build quasi-cyclic codes with designed minimum distance.

\begin{thm}\label{QC-Design}
Let $n$ and $\ell$ be positive integers with $\gcd(n,q)=1$ and $\ell<q$. Let $\cc$ be an $[\ell n,k,d]_q$ quasi-cyclic code generated by $\mathbf{G}^{\prime}(x)$ as in \eqref{Design1}. For each $i\in[\ell]$, let $\mathfrak{C}_i$ be the cyclic code of length $n$ over $\Ff_q$ with zero set $\bigcup_{j=1}^{m_i} E_{i,j}$. If $(L_i,d_{L_i})\in D(\mathfrak{C}_i)$ is a defining set bound given in Definition \ref{def1}, then
\begin{equation}\label{QCD}
d\geq \min\{d_{L_1},2 \, d_{L_2},\cdots,\ell \, d_{L_\ell}\}.
\end{equation}
\end{thm}

\begin{proof}
Letting $x^n-1 = f(x) \, \mathfrak{g}(x)$, we devise an $\ell\times \ell$ upper triangular matrix over $\Ff_q[x]$
\begin{equation}
\mathbf{H}(x)=\begin{pmatrix}
x^n-1&0&\cdots&0&-f(x) \, h_{1}(x)\\
0&x^n-1&\cdots&0&-f(x) \, h_{2}(x)\\
\vdots& \vdots &  \ddots & \vdots & \vdots\\
0&0&\cdots&x^n-1&-f(x) \, h_{\ell-1}(x)\\
0&0&\cdots&0&f(x)
\end{pmatrix}.
\end{equation}
It is clear that $\mathbf{H}(x) \, \mathbf{G}^{\prime}(x) = (x^n-1) \, \mathbf{I}_\ell$, which implies that $\mathbf{G}^{\prime}(x)$ generates a quasi-cyclic code $\cc$ of length $\ell n$ and eigenvalue set $E$. Theorem \ref{ISB} and Remark \ref{gen-matrix-remark} require us to determine the eigencode $\mathbb{C}_i$ of $\bigcap_{\beta\in L_i}V_{\beta}$ for each $i\in[\ell]$. We note that the eigenspace $V_\beta$ is
\[
V_{\beta}=\left\{\mathbf{u}\in\Ff_{q^r}^\ell:\mathbf{G}^{\prime}(\beta)\mathbf{u}^{\top}=\mathbf{0}^{\top}\right\}.
\]
We verify that $\bigcap_{\beta\in L_i}V_{\beta}$ is generated by the vector $(\gamma_1^i,\cdots,\gamma_{\ell-1}^i,1)$ for each $i\in[\ell]$. Hence, the correponding eigencode $\mathbb{C}_i$ of $\bigcap_{\beta\in L_i}V_{\beta}$ has a parity-check matrix $(\gamma_1^i,\cdots,\gamma_{\ell-1}^i,1)$. For each $t\in[\ell]$, the code $\bigcap_{i=1}^t\mathbb{C}_i$ has a parity-check matrix
\[
\begin{pmatrix}
\gamma_1&\gamma_2&\cdots&\gamma_{\ell-1}&1\\
\gamma_1^2&\gamma_2^2&\cdots&\gamma_{\ell-1}^2&1\\
\vdots& \vdots &  \ddots & \vdots\\
\gamma_1^t&\gamma_2^t&\cdots&\gamma_{\ell-1}^t&1\\
\end{pmatrix}.
\]
The minimum distance of $\bigcap_{i=1}^t\mathbb{C}_i$ is clearly $t+1$ for $t\in[\ell-1]$ and $\bigcap_{i=1}^\ell\mathbb{C}_i=\{\mathbf{0}\}$. Theorem \ref{ISB} allows us to conclude that $d(\cc) \geq \min\{d_{L_1}, 2 \, d_{L_2}, \cdots,\ell \, d_{L_\ell}\}$.
\end{proof}

\begin{example}\label{QC-designed}
Let $q=3$, $n=5$, $\ell=2<q$, and $\gamma=-1$. Considering $\mathfrak{g}(x)=x^5-1$, whose factorization over $\F_3$ is
\[
\mathfrak{g}(x)=x^5-1=\mathfrak{g}_{1}(x) \, \mathfrak{g}_{2}(x) = (x-1) \, (x^4+x^3+x^2+x+1),
\]
we infer that $m_1=m_2=1$. Let $\cc$ be the ternary quasi-cyclic code of length $10$ and index $2$ generated by the rows of
\[
\mathbf{G}^{\prime}(x)=\begin{pmatrix}
1 & h(x)\\
0 & \mathfrak{g}(x)
\end{pmatrix}.
\]
The eigenvalues of $\cc$ are $E=\Delta=\{\alpha,\alpha^2,\alpha^3,\alpha^4,1\}$, where $\alpha\in\F_{81}$ is a primitive $5^{\rm th}$ root of unity. Hence, $E=\bigcup_{i=1}^2 E_{i,1}$ such that $E_{1,1} = \{1\}$ and $E_{2,1} = \{\alpha,\alpha^2,\alpha^3,\alpha^4\}$.

We now construct $h(x)$. By (\ref{Polynomial-1}), we have
\begin{equation}\label{poly-h}
h(x) = \sum_{i=1}^2  a_{i}(x) \, \frac{\mathfrak{g}(x)}{\mathfrak{g}_{i}(x)}
= a_{1}(x) \, \mathfrak{g}_2(x) + a_{2}(x) \, \mathfrak{g}_1(x)
= a_{1}(x) \, (x^4+x^3+x^2+x+1)+ a_{2}(x) \, (x-1).
\end{equation}
The computation of the polynomials $a_{i}(x)$ goes as follows.
\begin{align*}
E_{1,1} = \{1\} & \implies a_{1}(1)=-(-1) \, (-1)^{-1} = -1 \mbox{ and}\\
E_{2,1} = \{\alpha,\alpha^2,\alpha^3,\alpha^4\} & \implies  a_{2}(\alpha^j)=-(-1)^2 \, (\alpha^j-1)^{-1} = -(\alpha^j-1)^{-1},  \mbox{\ for\ all\ } 1\leq j \leq 4.
\end{align*}
We conclude that $a_1(x)=-x$ and, by interpolation, obtain $a_2(x)=-x^3+x^2-1$. Substituting into (\ref{poly-h}) yields
\[
h(x)=-x^5+x^4+x^3+x^2+x+1,
\]
which implies
\[
\mathbf{G}^{\prime}(x)=\begin{pmatrix}
1 & -x^5+x^4+x^3+x^2+x+1\\
0 & x^5-1
\end{pmatrix}.
\]
We apply two defining set bounds on $E$ to estimate the minimum distance of $\cc$. First, we consider $(L_1,d_{L_1})=(\{\alpha,\alpha^2,\alpha^3\}, 4)$, obtained by the BCH bound. The corresponding eigencode is $\Cc_1$, which is generated by $\begin{pmatrix} 1 & -1 \end{pmatrix}$ over $\F_3$ with $d(\Cc_1)=2$. Second, we take $(L_2,d_{L_2})=(\{1\}, 2)$, again by using the BCH bound. The associated eigencode is $\Cc_2$, which is generated by $\begin{pmatrix} 1 & 1 \end{pmatrix}$ over $\F_3$ with $d(\Cc_2)=2$. We observe that $\Cc_1$ and $\Cc_2$ intersect trivially, yielding the estimate $\min\{d_{L_1},2d_{L_2}\}=4$. In fact, $\cc$ is a $[10,5,4]_3$ code and the improved spectral bound is sharp. \hfill $\blacksquare$
\end{example}

\begin{example}\label{Example1:QC-BCH}
When $n$ divides $q-1$, the eigenvalue set $E$ of the quasi-cyclic code in Theorem \ref{QC-Design} belongs to $\Ff_q$. If we take the consecutive set $L_i$ with $|L_i| = \left\lceil\frac{\delta-1}{i}\right\rceil$ for each $i\in[\ell]$ and $\sum_{i=1}^\ell|L_i|\leq n$ by using Theorem \ref{QC-Design}, then we get an $[\ell n, \ell n-\sum_{i=1}^\ell|L_i|, \geq\delta]_q$ quasi-cyclic code. Taking some $\ell$, $q$, and the consecutive sets $L_i$, we build some quasi-cyclic codes of length $\ell(q-1)$ with small Singleton defect. Table \ref{table-parameter} lists their parameters. The case $\ell=1$ corresponds to a BCH code.

\begin{table*}[ht!]
\caption{Parameters of Quasi-Cyclic Codes with Singleton defect $\leq 3$ .}
\label{table-parameter}
\renewcommand{\arraystretch}{1.2}
\centering
\begin{tabular}{cccccc}
\toprule
No. & $\ell$  & $q$ & $(|L_1|,|L_2|,|L_3|,|L_4|)$ & $[n,k,d]_q$ & Singleton defect\\
\midrule
$1$ & $2$ & $\geq 4$ & $(2,1,0,0)$ & $[2(q-1),2(q-1)-3,\geq 3]_q$ & $\leq 1$ \\
$2$ & $2$ & $\geq 5$ & $(3,1,0,0)$ & $[2(q-1),2(q-1)-4,\geq 4]_q$ & $\leq 1$ \\
$3$ & $2$ & $\geq 7$ & $(4,2,0,0)$ & $[2(q-1),2(q-1)-6,\geq 5]_q$ & $\leq 2$ \\
$4$ & $2$ & $\geq 9$ & $(5,2,0,0)$ & $[2(q-1),2(q-1)-7,\geq 6]_q$ & $\leq 2$ \\
$5$ & $2$ & $\geq 11$ & $(6,3,0,0)$ & $[2(q-1),2(q-1)-9,\geq 7]_q$ & $\leq 3$ \\
$6$ & $2$ & $\geq 11$ & $(7,3,0,0)$ & $[2(q-1),2(q-1)-10,\geq 8]_q$ & $\leq 3$ \\
\hline
$7$ & $3$ & $\geq 5$ & $(2,1,1,0)$ & $[3(q-1),3(q-1)-4,\geq 3]_q$ & $\leq 2$ \\
$8$ & $3$ & $\geq 7$ & $(3,1,1,0)$ & $[3(q-1),3(q-1)-5,\geq 4]_q$ & $\leq 2$ \\
$9$ & $3$ & $\geq 9$ & $(4,2,1,0)$ & $[3(q-1),3(q-1)-7,\geq 5]_q$ & $\leq 3$ \\
$10$ & $3$ & $\geq 9$ & $(5,2,1,0)$ & $[3(q-1),3(q-1)-8,\geq 6]_q$ & $\leq 3$ \\
\hline
$11$ & $4$ & $\geq 7$ & $(2,1,1,1)$ & $[4(q-1),4(q-1)-5,\geq 3]_q$ & $\leq 3$ \\
$12$ & $4$ & $\geq 7$ & $(3,1,1,1)$ & $[4(q-1),4(q-1)-6,\geq 4]_q$ & $\leq 3$ \\
\bottomrule
\end{tabular}
\end{table*}

Each cyclotomic coset from the eigenvalue set of our quasi-cyclic code has only one element, whereas the zero set of the $q$-ary BCH code with length $\ell n$ and designed minimum distance $\delta$ contains some cyclotomic cosets with at least two elements whenever $\ell n $ does not divide $(q-1)$. With this advantage, our coding scheme is easier to implement than that of the BCH codes. Given $11\leq q \leq 31$, $\ell=3$, and $n=q-1$, the dimensions of our quasi-cyclic codes are greater than those of the BCH codes, as shown in  Table \ref{table:QC-BCH}.

\begin{table}[ht!]
\caption{The dimensions and the actual minimum distances of $q$-ary BCH and quasi-cyclic codes of length $3(q-1)$ with designed minimum distance $6$.}
\label{table:QC-BCH}
\renewcommand{\arraystretch}{1.2}
\centering
\begin{tabular}{cc|cc|cc}
\toprule
$q$ & length & $dim$(BCH) & $dim$(QC) & $d$(BCH) & $d$(QC) \\
\midrule
$11$ & $30$ & $21$ & $22$ & $6$ & $6$\\
$13$ & $36$ & $23$ & $28$ & $6$ & $6$\\
$16$ & $45$ & $32$ & $37$ & $6$ & $6$\\
$17$ & $48$ & $39$ & $40$ & $6$ & $6$\\
$19$ & $54$ & $41$ & $46$ & $6$ & $6$\\
$23$ & $66$ & $57$ & $58$ & $6$ & $6$\\
$25$ & $72$ & $59$ & $64$ & $6$ & $6$\\
$29$ & $84$ & $75$ & $76$ & $6$ & $6$\\
$31$ & $90$ & $77$ & $82$ & $6$ & $6$\\
\bottomrule
\end{tabular}
\end{table}
\hfill $\blacksquare$
\end{example}

As Example \ref{Example1:QC-BCH} shows, the quasi-cyclic codes generated by Theorem \ref{QC-Design} perform better than the BCH codes under the given conditions. Next, we consider the performance of the improved spectral bound against the Jensen bound on the codes in Theorem \ref{QC-Design}.

We now give the Jensen bound for the quasi-cyclic code $\cc$ in Theorem \ref{QC-Design}. By definition, as a submodule of $R^\ell$ which is generated by the rows of the polynomial matrix $\mathbf{G}^{\prime}(x)$ in \eqref{Design1}, the code $\cc$ is spanned by
\begin{equation}\label{submodule-C}
\{(1,0,\cdots,0,h_1(x)),\cdots,(0,0,\cdots,1,h_{\ell}(x)),(0,0,\cdots,0,g(x))\}.
\end{equation}
Let $x^n-1=\mathfrak{g}(x)\prod_{b=1}^t f_b(x)$ with $f_b(x)\in\Ff_q[x]$ being an irreducible polynomial for each $b\in[t]$. Let $\mathbb{E}_{i,j}=\Ff_q[x]/\langle \mathfrak{g}_{i,j}(x)\rangle$ and $\mathbb{E}_{b}=\Ff_q[x]/\langle f_{b}(x)\rangle$ for each $i\in[\ell]$, $j\in[m_i]$, and $b\in[t]$. By \eqref{s22eq2} and \eqref{Decomposition}, we have
\[
R^\ell\cong\oplus_{i=1}^\ell\oplus_{j=1}^{m_i}\mathbb{E}_{i,j}^\ell \oplus_{b=1}^t \mathbb{E}_b^\ell \mbox{ and }
\cc\cong\oplus_{i=1}^\ell\oplus_{j=1}^{m_i}\cc_{i,j} \oplus_{b=1}^t \cc_b,
\]
where $\cc_{i,j}$ and $\cc_b$ are the constituent codes of $\cc$. Using \eqref{constituent}, \eqref{Polynomial-1}, and \eqref{submodule-C}, for each $i\in[\ell]$, $j\in[m_i]$, and $b\in[t]$, we obtain
\begin{align*}
\cc_{i,j} & = \left\{\left(a_1,\cdots,a_{\ell-1}, \sum_{z=1}^{\ell-1}a_z \gamma_z^i\right): a_1,\cdots,a_{\ell-1} \in \mathbb{E}_{i,j}\right\} \mbox{ and} \\
\cc_b &= \left\{\left(a_1,\cdots,a_{\ell}\right) : a_1, \cdots, a_{\ell} \in \mathbb{E}_{b}\right\}.
\end{align*}
It is easy to check that the respective minimum distances of $\cc_{i,j}$ and $\cc_b$ are $d(\cc_{i,j})=2$ and $d(\cc_{b})=1$ for each $i\in[\ell]$, $j\in[m_i]$, and $b\in[t]$. Let $\mathcal{D}_{i,j}=\langle\theta_{i,j}\rangle$ and $\mathcal{D}_{b}=\langle\theta_{b}\rangle$ be the minimal cyclic codes that are isomorphic to the respective fields $\mathbb{E}_{i,j}$ and $\mathbb{E}_{b}$ for each $i\in[\ell]$, $j\in[m_i]$, and $b\in[t]$. Since all of the constituent codes of $\cc$ are nontrivial, it follows from Theorem \ref{thmJensen} that the Jensen bound for $\cc$ takes into consideration
\begin{equation}\label{Jensen estimate}
d(\cc_{i_1,j_1}) \, d \left(\oplus_{i=1}^\ell \, \oplus_{j=1}^{m_i}\langle\theta_{i,j}\rangle \, \oplus_{b=1}^t \langle\theta_{b}\rangle \right) \mbox{ or }  d(\cc_{b}) \, d \left(\oplus_{i=1}^\ell \, \oplus_{j=1}^{m_i}\langle\theta_{i,j}\rangle \, \oplus_{b=1}^t \langle\theta_{b}\rangle \right).
\end{equation}
Thanks to $d \left(\oplus_{i=1}^\ell \oplus_{j=1}^{m_i}\langle\theta_{i,j}\rangle
\oplus_{b=1}^t \langle\theta_{b}\rangle \right) = 1$, we deduce that the quasi-cyclic code $\cc$ in Theorem \ref{QC-Design} admits the Jensen bound $d_J\leq 2$. We have thus proved the following result.

\begin{prop}\label{prop5.4}
Let $\cc$ be the quasi-cyclic code defined in Theorem \ref{QC-Design}. If $d_{Spec}$ and $d_J$ are the improved spectral bound given by Theorem \ref{QC-Design} and the Jensen bound of $\cc$, respectively, then $d_{Spec}\geq d_J$.
\end{prop}

\begin{example}\label{jensen-beat}
We consider the $[10,5,4]_3$ quasi-cyclic code $\cc$ in Example \ref{QC-designed}, which is generated by
\[
\mathbf{G}^{\prime}(x)=\begin{pmatrix}
1 & -x^5+x^4+x^3+x^2+x+1\\
0 & x^5-1
\end{pmatrix}.
\]
We recall that $x^5-1 = \mathfrak{g}_1(x) \, \mathfrak{g}_2(x) =(x-1) \, (x^4+x^3+x^2+x+1)$, with $\mathfrak{g}_{1} (1) = 0$ and $\mathfrak{g}_2(\alpha)=0$, where $\alpha$ is a primitive $5^{\rm th}$ root of unity in $\F_{81}$. Hence, $R\cong \F_3 \oplus \F_{81}$ and the respective constituents $\cc_1\subset \F_3^2$ and $\cc_2\subset \F_{81}^2$ are generated by
\[
G_1 = \begin{pmatrix} 1 & 1 \end{pmatrix} \mbox{ and }
G_2 = \begin{pmatrix} 1 & -1 \end{pmatrix}.
\]
Since $d(\cc_1) = d(\cc_2) = 2$, we can use (\ref{Jensen estimate}) to immediately conclude that $d_J=2$, whereas $d(\cc)=4$ was sharply estimated by the improved spectral bound. \hfill $\blacksquare$
\end{example}

\section{Connections to Locally Repairable Codes}\label{S6}

Coding schemes for large-scale distributed storage systems have been developed to catch up with the exponential growth of data. Examples of distributed storage systems in practice include Google's {\tt Bigtable}, Microsoft Azure's {\tt Blob Storage}, and Amazon's {\tt S3}. Erasure-correcting codes have been widely utilized in such systems to achieve two main objectives, namely to store data efficiently and to repair errors rapidly.

An early popular approach is to use an $[m,k,d]_q$ MDS code. It repairs failures in up to $d-1$ nodes by accessing the $k$ other nodes. To quickly repair failures in a small number of nodes, the MDS coding scheme requires the utilization of \emph{all} $k$ surviving nodes. The incurred cost is considered too high, since failures in a single node is the most frequent and the probability of simultaneous failures in two or more nodes decreases significantly as a function of the number of failed nodes.

An alternative coding scheme, first proposed in \cite{Gopalan2012}, offers a lower cost. The scheme consists of locally repairable codes (LRCs) where accessing $\rho \ll k$ surviving nodes is sufficient to repair a single node failure. To repair multiple node failures, Prakash {\it et al.} in \cite{Prakash2012} introduced the concept of linear codes with $(\rho,\delta)$-locality. Such codes allow us to repair each single node failure by accessing only $\rho$ surviving nodes in the presence of additional $\delta-2$ node failures, with $\delta \geq 2$.
\begin{defn}\label{LRC}
An $[m,k,d]_q$ linear code $\cc$ has $(\rho,\delta)$-locality if, for each symbol $c_i$ with $1\leq i\leq m$, there exists a subset $S_i \subseteq [m] = \{1,\ldots,m\}$ containing $i$ such that
\begin{enumerate}
\item $|S_i| \leq \rho + \delta-1$ and
\item the local code $\cc|_{S_i}$ has minimum distance ${\rm d}(\cc|_{S_i}) \geq \delta$.
\end{enumerate}
Here the local code $\cc|_{S_i}$ is constructed by removing the entries indexed by the set $[m] \setminus S_i$ in each codeword of $\cc$. The code $\cc$ is referred to as a $(\rho,\delta)$-LRC.
\end{defn}

Let $\lceil \cdot \rceil$ be the ceiling function. The Singleton-type bound for an $[m,k,d]_q$ linear code with locality $(\rho,\delta)$ was shown in \cite{Prakash2012} to be
\begin{equation}\label{LRCbound1}
d \leq m-k+1-\left(\left\lceil\frac{k}{\rho}\right\rceil-1\right) (\delta-1).
\end{equation}
The above bound is not tight for the case where the local codes are non-MDS. Subsequently, Grezet \textit{et al.} in \cite{Grezet2019} proved another bound that covers both the MDS and non-MDS local codes. It reads
\begin{equation}\label{LRCbound2}
d \leq m- \left\lceil\frac{k}{\kappa}\right\rceil\mathcal{G}(\kappa,\delta)+\mathcal{G}\left(\left\lceil\frac{k}{\kappa}\right\rceil\kappa-k+1,\delta\right),
\end{equation}
where $\kappa$ is the upper bound on the dimension of the local codes. Let $k_{\rm opt}^{(q)}(m,d)$ stand for the largest possible dimension of a linear code with length $m$ and minimum distance $d$. Let $\mathbb{Z}_{+}$ be the set whose elements are the positive integers and $0$. Taking the alphabet size into consideration, we have from \cite{Grezet2019} an alphabet-dependent bound that states
\begin{equation}\label{LRCbound3}
k\leq\min_{z\in \mathbb{Z}_{+}}\left\{z+k_{\rm opt}^{(q)}\left(m-(x+1) \, \mathcal{G}(\kappa,\delta) + \mathcal{G}(\kappa-y,\delta), d \right)\right\},
\end{equation}
where $\kappa$ is the upper bound on the dimension of the local codes, and $x,y \in \mathbb{Z}_{+}$ are such that $z=x\kappa+y$, with $0\leq y<\kappa$.

Based on the concatenated structure of quasi-cyclic codes, we prove that quasi-cyclic codes have $(\rho,\delta)$-locality when some conditions are met.
\begin{thm}\label{QC-LRC}
Let $\gcd(n,q)=1$ and let the polynomial $x^n-1$ factor into $t$ irreducible polynomials $f_1(x),\cdots,f_t(x) \in \Ff_q[x]$. Let $\langle\theta_{i}\rangle$ be the cyclic code with check polynomial $f_i(x)$. Let $\cc$ be an $[\ell n,k,d]_q$ quasi-cyclic code of index $\ell$. Let $I$ be a proper subset of $[t]$. If $\cc$ has the concatenated structure $\bigoplus_{i\in I} \langle\theta_{i}\rangle\, \square \, \cc_{i}$, with $\cc_i$ being a nonzero constituent for each $i$, then $\cc$ has $(n-\delta+1,\delta)$-locality with
\[
\delta=d\left(\bigoplus_{i\in I}\langle\theta_{i}\rangle\right).
\]
\end{thm}
\begin{proof}
We write $x^n-1=\prod_{i=1}^tf_i(x)$. Since $\langle\theta_{i}\rangle$ is the cyclic code with check polynomial $f_i(x)$, we obtain, by \eqref{s22eq3},
\[
\cc=\bigoplus_{i\in I}\langle\theta_{i}\rangle\, \square \, \cc_{i}
=\left\{\left(\bigoplus_{i\in I}\psi_{i}(c_{i,1}),\cdots,\bigoplus_{i\in I}\psi_{i}(c_{i,\ell})\right):(c_{i,1},\cdots,c_{i,\ell})\in\cc_{i},i\in I\right\}.
\]
We note that $\psi_{i}(c_{i,1})\in\langle\theta_{i}\rangle$ for each $i\in I$. If $S_j=\{(j-1)n+1,\cdots,jn\}$ for $j\in[\ell]$, then the local code $\cc|_{S_j}$ is $\bigoplus_{i\in I}\langle\theta_{i}\rangle$. Since $I$ is a proper subset of $[t]$, the cyclic code $\bigoplus_{i\in I}\langle\theta_{i}\rangle$ has at least one zero. Hence, the minimum distance of the local code $\cc|_{S_j}$ is greater than $1$. By Definition \ref{LRC}, the quasi-cyclic code $\cc$ has $(n-\delta+1,\delta)$-locality.
\end{proof}

By Theorem \ref{QC-LRC}, a quasi-cyclic code $\cc=\bigoplus_{i\in I} \langle\theta_{i}\rangle\, \square \, \cc_{i}$ is an LRC if $I$ is a proper subset of $[t]$. This sufficient condition comes from the concatenated structure of quasi-cyclic codes. Starting from a generator polynomial matrix $\mathbf{G}^{\prime}(x)$ of $\cc$, we provide another sufficient condition for $\cc$ to be an LRC. We note that $I$ is a proper subset of $[t]$ if and only if there exists at least one constituent $\cc_{i}=\{\mathbf{0}\}$. Let the polynomial $x^n-1=\prod_{i=1}^tf_i(x)$ factor into $t$ irreducible polynomials in $\Ff_q[x]$. We recall that $r$ is the smallest positive integer such that $q^r\equiv 1 \pmod n$ and $\alpha$ is a primitive $n^{\rm th}$ root of unity in $\Ff_{q^r}$. Since the roots of each $f_i(x)$ are powers of $\alpha$, we let $v_i$ be the smallest nonnegative integer such that $f_i(\alpha^{v_i})=0$. By \eqref{constituent}, the constituent $\cc_{i}$ is the $\mathbb{E}_i$-span of the rows of $\mathbf{G}^{\prime}(\alpha^{v_i})$ with $\mathbb{E}_i=\Ff_q[x]/\langle f_i(x)\rangle$ for each $i\in[t]$. Hence, $\cc_{i}=\{\mathbf{0}\}$ if and only if $\mathbf{G}^{\prime}(\alpha^{v_i})$ is the $\ell\times\ell$ zero matrix, which implies that $\mathbf{G}^{\prime}(x) = f_i(x) \, \mathbf{G_1}^{\prime}(x)$ for some polynomial matrix $\mathbf{G_1}^{\prime}(x)$ over $\Ff_q[x]$. If a quasi-cyclic code $\cc$ has a generator polynomial matrix $\mathbf{G}^{\prime}(x)=\prod_{j=1}^g f_{i_j}(x) \, \mathbf{G_1}^{\prime}(x)$, then $\cc$ has $(n-\delta+1,\delta)$-locality with
\begin{equation}\label{LRC-Spectral}
\delta=d\left(\bigoplus_{i\in [t]\setminus\{i_1,\cdots,i_g\}}\langle\theta_{i}\rangle\right).
\end{equation}
The following corollary to Theorem \ref{QC-LRC} reveals several bounds for quasi-cyclic codes.

\begin{cor}\label{QC-LRC-Bound}
Let $\gcd(n,q)=1$ and let the polynomial $x^n-1$ factor into $t$ irreducible polynomials in $\Ff_q[x]$. Let $\cc$ be an $[\ell n,k,d]$ quasi-cyclic code of index $\ell$ having concatenated structure $\bigoplus_{i\in I} \langle\theta_{i}\rangle\, \square \, \cc_{i}$, with $I\subseteq [t]$. If $\delta=d\left(\bigoplus_{i\in I}\langle\theta_{i}\rangle\right)$ and $\kappa =
k_{\rm opt}^{(q)}(n,\delta)$, then the parameters $\ell n$, $k$, and $d$ satisfy the following three statements.
\begin{enumerate}
\item $d \leq \ell n-k+1- \left(\left\lceil\frac{k}{n-\delta+1}\right\rceil-1\right)(\delta-1)$,
\item $d\leq \ell n- \left\lceil\frac{k}{\kappa}\right\rceil \,  \mathcal{G}(\kappa,\delta)+\mathcal{G}\left(\left\lceil\frac{k}{\kappa}\right\rceil\kappa-k+1,\delta\right)$,
\item $k\leq\min_{z\in \mathbb{Z}_{+}} \left\{z+k_{\rm opt}^{(q)} \left(\ell n-(x+1) \, \mathcal{G}(\kappa,\delta) + \mathcal{G}(\kappa-y,\delta),d\right)\right\}$, with $x,y\in \mathbb{Z}_{+}$ and $0\leq y<\kappa$ defining $z=x\kappa+y$.
\end{enumerate}
\end{cor}

\begin{proof}
Using \eqref{LRCbound1}, \eqref{LRCbound2}, \eqref{LRCbound3}, and Theorem \ref{QC-LRC}, the desired results follow whenever $\cc=\bigoplus_{i\in I} \langle\theta_{i}\rangle\, \square \, \cc_{i}$ and $I$ is a proper subset of $[t]$. If $\cc=\bigoplus_{i=1}^t \langle\theta_{i}\rangle\, \square \, \cc_{i}$, then the same method used in proving Theorem \ref{QC-LRC} yields $\delta = d\left(\bigoplus_{i=1}^t\langle\theta_{i}\rangle\right)=1$, ensuring that the above three statements hold.
\end{proof}

\begin{remark}
The value of $\delta$, evaluated based on the concatenated structure, is essential in the three bounds in Corollary \ref{QC-LRC-Bound}. Given an $[\ell n,k,d]_q$ random quasi-cyclic code $\cc$ of index $\ell$ with codewords in the form of \eqref{s2eq1}, each column in each $n\times \ell$ array of $\cc$ belongs to a local code of length $n$ and minimum distance $\delta$. That is to say, the value of $\delta$ follows from the minimum distance of the code generated by the first column in each $n\times \ell$ array of $\cc$.
\end{remark}

Luo, Ezerman, and Ling in \cite{Luo2022} constructed $q$-ary quasi-cyclic codes of index $\ell \leq q$ meeting the bound in Statement $(1)$ of Corollary \ref{QC-LRC-Bound}. Their approach utilizes \textit{matrix-product codes} of the form
\[
\left\{\left(\sum_{i=1}^M \cF_i \,
a_{i,1},\, \ldots, \, \sum_{i=1}^M \cF_i \, a_{i,N}\right) \, :  \cF_1\in\cc_1, \, \ldots, \, \cF_M \in \cc_M \vphantom{} \right\},
\]
where each codeword consists of $N$ component row vectors, each of length $n$. Matrix-product codes form a special class of quasi-cyclic codes of index $N$ if their short component codes $\cc_i$ are cyclic codes for all $i \in [M]$. This insight leads to an explicit construction of $q$-ary quasi-cyclic codes of index $q+1$ that achieve the bound in Statement $(1)$ of Corollary \ref{QC-LRC-Bound}.

We label the elements of $\Ff_q$ as $\gamma_1$, $\cdots$, $\gamma_q$ and let $v_1,\cdots,v_q$ be nonzero elements of $\Ff_q$. An extended \emph{generalized Reed-Solomon} (GRS) code is generated by
\[
\left(\begin{matrix}
v_1 & v_2 &\cdots & v_q&0&\\
v_1\gamma_1 & v_2\gamma_2 &\cdots & v_q\gamma_q&0&\\
\vdots& \vdots &  \ddots & \vdots&\vdots&\\
v_1\gamma_1^{k-1} & v_2\gamma_2^{k-1} &\cdots & v_q\gamma_q^{k-1}&1&\\
\end{matrix}\right).
\]
We know from \cite[Page. 209]{Ling2004} that such a code is MDS with parameters $[q+1,k,q-k+2]_q$.

\begin{thm}\label{LRC-C1}
Let $n$ divide $(q-1)$ and let $x^n-1=\prod_{i=0}^{n-1}(x-\beta^i)$, with $\beta$ being a primitive $n^{\rm th}$ root of unity in $\Ff_q$. Let $\langle\theta_{i}\rangle$ be the cyclic code of length $n$ with check polynomial $x-\beta^i$ for $i\in[n-\delta+1]$ and $2\leq\delta< n$. Let $\cc_1$ be an extended GRS code with parameters $[q+1,q-a+2,a]_q$ and let $\cc_i$ be a $[q+1,q-a+1,a+1]_q$ extended GRS code for $i \in \{2,\cdots,n-\delta+1\}$. If $a\leq \frac{\delta}{n-\delta}$, then the quasi-cyclic code $\cc=\bigoplus_{i=1}^{n-\delta+1}\langle\theta_{i}\rangle\, \square \, \cc_{i}$ has parameters $[(q+1)n,(q+1-a)(n-\delta+1)+1,a \, n]_q$ and achieves equality in both bounds in Statements $(1)$ and $(2)$ of Corollary \ref{QC-LRC-Bound}.
\end{thm}
\begin{proof}
By the BCH bound, the code $\bigoplus_{i=1}^{z}\langle\theta_{i}\rangle$ has parameters $[n,z,n-z+1]_q$ and is MDS for each $z\in[n-\delta+1]$. By Theorem \ref{QC-LRC}, the code $\cc$ has $(n-\delta+1,\delta)$-locality. By the way $\cc$ is defined, it has length $(q+1) \, n$ and dimension $(q-a+1)(n-\delta+1)+1$. By Theorem \ref{thmJensen}, $d(\cc) \geq \min\{a \, n,(a+1) \, \delta\}$. Since $a \leq \frac{\delta}{n-\delta}$, we have $d \geq a \, n$. By the bound in Statement $(1)$ of Corollary \ref{QC-LRC-Bound}, we derive
\[
d(\cc) \leq (q+1)n-((q-a+1)(n-\delta+1)+1)+1 \\ -\left(\left\lceil\frac{(q-a+1)(n-\delta+1)+1}{n-\delta+1}\right\rceil-1 \right) (\delta-1) = a \, n.
\]
We verify that the local code $\bigoplus_{i=1}^{n-\delta+1}\langle\theta_{i}\rangle$ is an $[n,n-\delta+1,\delta]_q$ MDS code, which implies $\kappa=k_{\rm opt}^{(q)}(n,\delta)=n-\delta+1$. By the bound in Statement $(2)$ of Corollary \ref{QC-LRC-Bound}, we get
\begin{multline*}
d(\cc) \leq (q+1) \, n -\left\lceil\frac{(q-a+1)(n-\delta+1)+1}{n-\delta+1}\right\rceil\mathcal{G}(n-\delta+1,\delta) \\
 +\mathcal{G}\left(\left\lceil\frac{(q-a+1)(n-\delta+1)+1}{n-\delta+1}\right\rceil(n-\delta+1)-(q-a+1)(n-\delta+1),\delta\right)\\
= (q+1) \, n-(q-a+2) \, n+ n =a \, n.
\end{multline*}
Thus, the quasi-cyclic code $\cc$ has parameters $[(q+1) \, n,(q+1-a)(n-\delta+1)+1,a \, n]_q$ and attains the equality in both bounds in Statements $(1)$ and $(2)$ of Corollary \ref{QC-LRC-Bound}.
\end{proof}

We now propose the second construction by introducing the factorization of $x^n-1$ whenever $n$ divides $(q+1)$. Let $\mu$ be a primitive $n^{\rm th}$ root of unity in $\Ff_{q^2}$. If $n$ is odd, then
\[
x^n-1 =(x-1) \, f_1(x) \cdots \, f_{\frac{n-1}{2}}(x),
\]
where $f_i(x) = x^2-(\mu^i+\mu^{-i}) \, x + 1$ is an irreducible polynomial over $\Ff_q$ for each $i\in[\frac{n-1}{2}]$. If $n$ is even, then
\[
x^n-1= (x-1) \, (x-\mu^{\frac{n}{2}}) \, f_1(x) \cdots \, f_{\frac{n-2}{2}}(x).
\]

\begin{thm}\label{LRC-C2}
Let $n$ divide $(q+1)$ and let $\delta$ be a positive integer such that $n-\delta$ is even and $3\leq\delta\leq \frac{n-1}{2}$. Let $\langle\theta_{1}\rangle$ be the cyclic code of length $n$ with check polynomial $x-1$. Let $\langle\theta_{i}\rangle$ be the cyclic code of length $n$ with check polynomial $f_{i-1}(x)$ for $i\in[\frac{n-\delta}{2}]$. Let $\cc_1$ be an extended GRS code with parameters $[q+1,q-a+2,a]_q$ and let $\cc_i$ be a $[q+1,q-a+1,a+1]_{q^2}$ extended GRS code for $i=2,\cdots,\frac{n-\delta}{2}$. If $a\leq \frac{\delta}{n-\delta}$, then the quasi-cyclic code $\cc=\bigoplus_{i=1}^{n-\delta+1}\langle\theta_{i}\rangle\, \square \, \cc_{i}$ has parameters $[(q+1) \, n,(q+1-a)(n-\delta+1)+1,a \, n]_q$ and achieves equality in both bounds in Statements $(1)$ and $(2)$ of Corollary \ref{QC-LRC-Bound}.
\end{thm}

\begin{proof}
The required conclusion follows from a similar argument as the one that proves Theorem \ref{LRC-C1}.
\end{proof}

\begin{remark}
A construction of quasi-cyclic codes meeting equality in both bounds in Statements $(1)$ and $(2)$ of Corollary \ref{QC-LRC-Bound} can already be found in \cite{Luo2022}. The codes have index $\ell \leq q$ and parameters
\[
[\ell \, n,(\ell-a)(n-\delta+1)+1,a \, n]_q,
\]
with $n$ dividing $(q+1)$ or $n$ dividing $(q-1)$. Theorems \ref{LRC-C1} and \ref{LRC-C2} can produce optimal quasi-cyclic codes with index $q+1$ and lengths up to $(q+1)^2$. Thus, many of our codes have new parameters since the lengths of previously known codes from \cite{Luo2022} are only up to $q^2+q$.
\end{remark}

The above two constructions of optimal quasi-cyclic codes are based on their concatenated structure and the Jensen bound. At the end of this section, we propose another construction of optimal quasi-cyclic codes based on the improved spectral bound in Theorem \ref{ISB}. We start by designing the generator polynomial matrix.

Let $n$ divide $(q-1)$ and let $x^n-1=\prod_{i=0}^{n-1}(x-\beta^i)$, with $\beta$ being a primitive $n^{\rm th}$ root of unity in $\Ff_q$. Let us denote by
$\Delta=\{\alpha^i:0\leq i\leq n-1\}$ the set of roots of $x^n-1$. We define the polynomial $h(x)\in\Ff_q[x]$ as
\[
h(x) = \sum_{i=0}^{n-1} a_{i}(x) \, \frac{x^n-1}{x-\beta^i} \mbox{, where }
a_{i}(\beta^i) = -\gamma \left(\prod_{j=0,j\neq i}^{n-1}(\beta^i-\beta^j)\right)^{-1}
\]
for each $i \in \{0,\cdots,n-1\}$, and $\gamma\in\Ff_q\setminus\{0\}$. Let $g(x)=\prod_{i=1}^{n-1}(x-\beta^i)$. We then construct, over $\Ff_q[x]$, the $\ell\times\ell$ upper-triangular matrix 
\begin{equation}\label{LRC-C3-G}
\mathbf{G}^{\prime}(x)=\begin{pmatrix}
g(x)&0&\cdots&0&g(x)h(x)\\
0&g(x)&\cdots&0&g(x)h(x)\\
\vdots& \vdots &  \ddots & \vdots & \vdots\\
0&0&\cdots&g(x)&g(x)h(x)\\
0&0&\cdots&0&x^n-1
\end{pmatrix}.
\end{equation}

\begin{thm}\label{LRC-C3}
Let $n$ divide $(q-1)$ and let $\ell>1$ be a positive integer. The quasi-cyclic code $\cc$ of index $\ell$ generated by \eqref{LRC-C3-G} has parameters $[\ell \, n,\ell-1,2 n]_q$ and attains equality in both bounds in Statements $(1)$ and $(2)$ of Corollary \ref{QC-LRC-Bound}.
\end{thm}

\begin{proof}
Let
\[
\mathbf{H}(x) =
\begin{pmatrix}
x-1&0&\cdots&0&-h(x)\\
0&x-1&\cdots&0&-h(x)\\
\vdots& \vdots &  \ddots & \vdots & \vdots\\
0&0&\cdots&x-1&-h(x)\\
0&0&\cdots&0&1
\end{pmatrix}.
\]
It is clear that $\mathbf{H}(x) \, \mathbf{G}^{\prime}(x) = (x^n-1) \, \mathbf{I}_\ell$. The code $\cc$ formed by \eqref{LRC-C3-G} is a quasi-cyclic code of index $\ell$. Since $\mathbf{G}^{\prime}(x) = g(x) \, \mathbf{G}_1^{\prime}(x)$, it follows from \eqref{LRC-Spectral} that $\cc$ has $(1,n)$-locality. Based on both bounds in Statements $(1)$ and $(2)$ of Corollary \ref{QC-LRC-Bound}, we deduce that $d(\cc)\leq 2n$.

Since $\cc$ is the quasi-cyclic code with generator polynomial matrix $\mathbf{G}^{\prime}(x)$, the eigenvalue set of $\cc$ is $\Delta$. Let $\mathfrak{C}$ be the cyclic code of length $n$ over $\Ff_q$ with zero set $\Delta$. To utilize the improved spectral bound in Theorem \ref{ISB}, we let $L_1=\Delta$ and $L_2=\{\beta^i:i\in[n-1]\}$. Their respective defining set bounds according to Definition \ref{def1} are $(L_1,d_{L_1}=\infty)$ and $(L_2,d_{L_2}=n)$. It is immediate to confirm that the common eigenspace $\bigcap_{\beta\in L_1}V_{\beta}$ is generated by the vector $(\gamma,\cdots,\gamma,1)$ and $\bigcap_{\beta\in L_2}V_{\beta}$ is the full space $\Ff_q^\ell$. The corresponding eigencodes $\mathbb{C}_1$ and $\mathbb{C}_2$ have respective parameters $[\ell,\ell-1,2]_q$ and $[\ell,0,\infty]_q$. By Theorem \ref{ISB}, we have $d(\cc)\geq 2n$. Hence, $\cc$ has parameters $[\ell n,\ell-1,2n]_q$ and attains equality in both bounds in Statements $(1)$ and $(2)$ of Corollary \ref{QC-LRC-Bound}.
\end{proof}

\begin{remark}
The co-indices of the quasi-cyclic codes in both Theorem \ref{LRC-C1} and Theorem \ref{LRC-C3} divide $(q-1)$. The parameters of the codes in Theorem \ref{LRC-C3} differ from those in Theorem \ref{LRC-C1}. It is immediate to confirm that Theorem \ref{LRC-C1}, unlike Theorem \ref{LRC-C3}, cannot generate optimal quasi-cyclic codes with parameters $[\ell n, \ell-1, 2n]_q$ due to $2\leq \delta<n$.
\end{remark}

\begin{example}\label{QC-local}
Let $q=5$, $n=4=q-1$, $\ell=3$, and $\gamma=-1$. We consider the factorization of $x^4-1$ over $\F_5$
\[
x^4-1=(x-1) \, (x-\beta) \, (x+1) \, (x+\beta),
\]
where $\beta\in\F_{5}$ is a primitive $4^{\rm th}$ root of unity. We build a quinary quasi-cyclic code $\cc$ of length $12$ and index $3$ generated by the rows of
\[
\mathbf{G}^{\prime}(x)=\begin{pmatrix}
g(x) & 0 & g(x)h(x)\\
0 & g(x) & g(x)h(x)\\
0 & 0 & x^4-1
\end{pmatrix}
\]
such that $g(x)=(x-\beta) \, (x+1) \, (x+\beta) = x^3 + x^2 + x + 1$, whereas $h(x)$ is of the form
\begin{align}\label{poly-h-2}
h(x) &= \sum_{i=0}^3  a_{i}(x) \, \frac{x^4-1}{x-\beta^i}\notag\\
&= a_{0}(x) \, (x^3+x^2+x+1)+ a_{1}(x) \, (x^3+\beta x^2-x-\beta) + a_{2}(x) \, (x^3-x^2+x-1) + a_{3}(x) \, (x^3-\beta x^2-x+\beta).
\end{align}
The polynomials $a_{i}(x)$ satisfy
\begin{align*}
a_{0}(1) &=\left(\prod_{j=1}^{3}(1-\beta^j)\right)^{-1} = -1,
& a_{1}(\beta) &=\left(\prod_{j=0,j\neq 1}^{3}(\beta-\beta^j)\right)^{-1} = -\beta, \\
a_{2}(-1) &=\left(\prod_{j=0,j\neq 2}^{3}(-1-\beta^j)\right)^{-1} = 1,
& a_{3}(-\beta) &=\left(\prod_{j=0}^{2}(-\beta-\beta^j)\right)^{-1} = \beta.
\end{align*}
Hence, $a_i(x)=-x$, for each $i \in \{0,\ldots,4\}$, and substituting into (\ref{poly-h-2}) yields $h(x)=x^4$, giving us
\[
\mathbf{G}^{\prime}(x)=\begin{pmatrix}
x^3 + x^2 + x + 1 & 0 & x^7 + x^6 + x^5 + x^4\\
0 & x^3 + x^2 + x + 1 & x^7 + x^6 + x^5 + x^4\\
0 & 0 & x^4-1
\end{pmatrix}.
\]
The resulting quasi-cyclic code $\cc$ is a $[12,2,8]_5$ code with $(1,4)$-locality.\hfill $\blacksquare$
\end{example}

\section{CONCLUDING REMARKS}\label{S7}

Spectral bounds can be determined from generator polynomial matrices of quasi-cyclic codes. We have shown that the improved version of the bounds, stated in Theorem \ref{ISB}, is now the lower bound to beat in terms of performance. Previous spectral bounds had been demonstrated to be generally inferior to the Jensen bound. Our improved spectral bound refines the spectral bound in \cite{Ezerman2021}, which we reproduce here as Theorem \ref{thm21}. Numerical results on the bounds and the computed actual minimum distances of randomly generated binary and ternary quasi-cyclic codes show that the improved spectral bound equals or outperforms the Jensen bound on almost all occasions, even when we only consider a restricted number $k=4$ of defining set bounds and a few terms $s \in \{2,3,4\}$ in the bound. The performance of the improved spectral bound might get better if we take into account more terms and more defining set bounds. This is strongly supported by the structural analysis in Section \ref{S3}, as highlighted by Remark \ref{remark1}.

Based on our improved spectral bound, it also becomes easier to construct good quasi-cyclic codes with designed minimum distances. Prior to this work, such a construction requires some strict conditions to be met. In prior constructions, the designed distance of a quasi-cyclic code of length $\ell n$ is upper bounded by the co-index $n$, for any given index $\ell$. The improved spectral bound allows us to drop the restriction on the eigencodes in Theorem \ref{thm21}.

In terms of dimensions, the quasi-cyclic codes constructed based on Theorem \ref{QC-Design} perform better than the BCH codes of comparative length and designed distance. We provide a proof that the improved spectral bound in Theorem \ref{QC-Design} beats the Jensen bound. The proof of Theorem \ref{LRC-C3} gives a construction of quasi-cyclic codes of length $\ell n$ with designed minimum distance $2n$. This construction yields quasi-cyclic codes whose designed distances exceed the co-index $n$. Our improved bound leads to a potent tool to construct excellent quasi-cyclic codes by designing their respective minimum distances judiciously.

The flexible structure of the improved spectral bound makes it rather hard to find conditions that ensure the improved spectral bound \emph{always performs better} than the Jensen bound. A particularly interesting direction is to evaluate their behavior by using theoretical techniques in the general case. In Section \ref{S6}, we have given a sufficient and necessary condition for a constituent to be zero. On the other hand, if all constituents of a quasi-cyclic code are nonzero, then the Jensen bound is upper bounded by its index. Example \ref{jensen-beat} illustrates this fact and the improved spectral bound gives a strictly better estimate. Based on this observation, a more elaborate analysis on the quasi-cyclic codes, which have nonzero constituents and a more general polynomial generator matrix than (\ref{Design1}), should be carried out.

% Generated by IEEEtran.bst, version: 1.14 (2015/08/26)

\end{document}